\documentclass[journal,comsoc]{IEEEtran}
\usepackage{graphicx}
\usepackage{url}
\usepackage{amsmath,amsthm}
\usepackage{mathtools}
\usepackage{bm}
\usepackage{acronym}
\usepackage{mathrsfs}
\usepackage{color}
\usepackage{float}
\usepackage{dblfloatfix}
\usepackage{color}
\usepackage{subcaption}

\newcommand{\beq}{\begin{equation}}
\newcommand{\eeq}{\end{equation}}
\newcommand{\beqa}{\begin{eqnarray}}
\newcommand{\eeqa}{\end{eqnarray}}

\newtheorem{theorem}{Theorem}[section]

\newtheorem{lemma}[theorem]{Proposition}

\usepackage[cmintegrals]{newtxmath}

\begin{document}

\title{Statistical Assessment of IP Multimedia Subsystem in a Softwarized Environment: \\ a Queueing Networks Approach}
\author{Mario~Di~Mauro,~\IEEEmembership{Member,~IEEE,}
	Antonio~Liotta,~\IEEEmembership{Senior Member,~IEEE}
	\IEEEcompsocitemizethanks{\IEEEcompsocthanksitem M. Di Mauro is with the Department of Information and Electrical Engineering and Applied Mathematics (DIEM), University of Salerno, 84084, Fisciano, Italy (E-mail: mdimauro@unisa.it).
		\newline
		
		A. Liotta is with the School of Computing, Edinburgh Napier University, Edinburgh EH105DT , U.K. (E-mail: a.liotta@napier.ac.uk)
		
	}
}

\maketitle

\begin{abstract}

The Next Generation 5G Networks can greatly benefit from the synergy  between virtualization paradigms, such as the Network Function Virtualization (NFV), and service provisioning platforms such as the IP Multimedia Subsystem (IMS). The NFV concept is evolving towards a lightweight solution based on {\em containers} that, by contrast to classic virtual machines, do not carry a whole operating system and result in more efficient and scalable deployments. On the other hand, IMS has become an integral part of the 5G core network, for instance, to provide advanced services like Voice over LTE (VoLTE). In this paper we combine these virtualization and service provisioning concepts, deriving a containerized IMS infrastructure, dubbed cIMS, providing its assessment through statistical characterization and experimental measurements. Specifically, we: $i)$ model cIMS through the \textit{queueing networks} methodology to characterize the utilization of virtual resources under constrained conditions; $ii)$ draw an extended version of the  Pollaczek-Khinchin formula, which is useful to deal with bulk arrivals; $iii)$ afford an optimization problem focused at maximizing the whole cIMS  performance in the presence of capacity constraints, thus providing new means for the service provider to manage service level agreements (SLAs); $iv)$ evaluate a range of cIMS scenarios, considering different queuing disciplines including also multiple job classes. An experimental testbed based on the open source platform \textit{Clearwater} has been deployed to derive some realistic values of key parameters (e.g. arrival and service times).

\end{abstract}


\IEEEpeerreviewmaketitle

\section{Introduction}

\IEEEPARstart{S}{oftwarization} plays a crucial role in 5G network infrastructures \cite{kellersurvey,basta2017}. It refers to those systems, tools, and procedures which intervene across the transformation process at the basis of novel telecommunication frameworks. The Network Function Virtualization (NFV) paradigm plays a central role in this process, since it provides a series of advantages such as flexibility in service provisioning, efficiency in resource utilization, and considerable potential for cost reductions \cite{nfvsurvey,nfvsurvey2}. Virtualized environments have revolutionized the deployment of new services by means of the so called Service Function Chains (SFC), which allow a smart and customizable composition of $5$G-based network functions \cite{sfclatre,sfcturck}, and open the door to new strategies for resource allocation \cite{optalloc} along a more efficient management of distributed infrastructures \cite{marquez}. 

An interesting evolution of virtualized systems is represented by container-based network architectures \cite{latrecontainer,boutaba}. Unlike classic virtual machines, containers are lightweight software instances which do not embed a whole operating system (OS). Containers run on the same hardware by sharing the OS that is mounted on the physical machine, thus, the isolation is guaranteed at the OS process level \cite{containersurvey18,rehman}. These processes are managed through dedicated platforms such as Docker \cite{docker}, typically composed of a main engine (often referred to as the container manager) and of a certain number of instances that can be easily deployed across a different set of cloud environments. 
Moreover, container technology is particularly suited to implement the network slicing concepts, providing a unique opportunity to assign fully dedicated resources per slice, which can in turn be dynamically reassigned to boost the cost/efficiency trade-off of the whole system \cite{marquez}.

Because of this level of versatility, container technologies are attracting the attention of the Telco industry, who see great value in dynamic transportation and efficient execution. Exemplary is the case of AT\&T that has been one of the first to expose (on a dedicated platform) small, independent, and self-contained business functions through container-based APIs \cite{att}. 

Another core part of 5G infrastructures is the IP Multimedia Subsystem (IMS), which has been identified as the best candidate for delivering multimedia content and services \cite{camarillo} such as gaming, presence, and Peer-to-Peer resource sharing \cite{liotta07}. IMS is also well suited for virtualized/containerized deployments \cite{survey16,nguyen18}, which is why it is drawing the attention of industry top players \cite{eri2014,nec2015}. The ETSI standardization group has included the virtualized IMS framework as a desirable solution for mobile next generation networks \cite{etsi}. In fact, the virtualized IMS can be considered a particular realization of an SFC, since the softwarized nodes have to be traversed in a predetermined order to provide specific services (e.g. IMS Registration).  

The versatility of a virtualized IMS solution is further amplified within the Clearwater project \cite{clearwater}, an open-source IMS implementation (written in Java and C++) deployable on a container-based architecture that represents a valuable example of a softwarized network infrastructure \cite{cotroneo17,cotroneo17bis}. Remarkably, containerized IMS functionalities offered by Clearwater have been embodied in a Proof-of-Concept pilot by Norwegian telco provider Telenor \cite{telenor}, where Red Hat Openshift has been exploited as container platform. 

Inspired by this fruitful combination between virtualization and service provisioning concepts, in this paper we consider a container-based IMS framework, dubbed cIMS. We carry out a statistical characterization under a range of scenarios, where some realistic parameters are directly derived by a pilot implementation on the Clearwater platform.

Our assessment relies on the queueing networks methodology which has a double virtue: on one hand, it is a well-assessed framework that allows to capture the behavior of interconnected systems (such as the case of cIMS nodes); on the other hand, it represents the most appropriate theory to characterize cases in which the resource usage is constrained by a wait, as often occurs in virtualized environments where it is necessary to share resources. 

Our modeling phase (which also embodies a generalization of Pollaczek-Khinchin formula for bulk requests) is preparatory to afford two analyses. The first one concerns a performance evaluation of different cIMS deployments, whereby capacity constraints are introduced, which requires solving a convex optimization problem. The second one is aimed at evaluating different cIMS scenarios by taking into account two formalisms: the Jackson framework \cite{jackson63}, useful to model networks nodes obeying to First-Come-First-Serve (FCFS) queueing discipline and where a single type of job is admitted; the BCMP framework \cite{bcmp75}, where nodes can implement disciplines other than the classic FCFS, and where multiple types of jobs are permitted. Results of aforementioned analyses reveal how the cIMS performance is affected either by capacity constraints and by deployment scenarios, offering to telco providers helpful indications for SLA tuning.

The rest of paper is structured as follows. Section \ref{sec:rw} provides an {\em excursus} of works that afford similar approaches, leading to highlighting the main contribution of our work in relation to the existing literature. Section \ref{sec:ims} is aimed at describing the Clearwater framework as a way to realize IMS platforms, which is the basis for our cIMS implementation. Section \ref{sec:qmodel} introduces the adopted queueing networks model, where we consider the case of bulk arrivals, and describe the optimization problem. In Section \ref{sec:exp}, we afford a performance analysis, by considering several conditions of deployments (e.g. single/multiple class requests). Finally, Section \ref{sec:concl} draws conclusion and provides hints for future research.

\section{Related Research and Contributions}
\label{sec:rw}

Over the recent years, academia and industry alike have devoted an increasing interest to the characterization of 5G network architectures and their constitutive elements, with analyses ranging from optimal resource distribution of virtualized multimedia nodes \cite{duan17} to availability characterization of virtualized IMS deployments \cite{dimaurotsc}. Yet, research and practical developments in this area are incredibly fast-paced, and it would take a dedicated review paper to provide a comprehensive snapshot. Instead, in this section we focus on recent works that have closer relevance or affinity to our contributions. 

In many cases, existing works embed a theoretical modeling of novel network infrastructures but fall short on experimental part, due to the difficulty in developing practical IMS implementations. We overcome this limitation, providing both theoretical and experimental results.

We adopt a queueing theory approach, which has been profitably exploited in some recent works to face various issues relating to modern network architectures. This is the case of \cite{xiong16}, where the authors propose performance models for OpenFlow switches and SDN controllers, respectively as $M^X/M/1$ and $M/G/1$ queueing systems. They also carry out a numerical analysis in a simulated environment, using the Cbench stress test tool. A similar analysis has been afforded in \cite{sood16}, where the authors model SDN switches by exploiting $M/Geo/1$ queues, assuming service times that obey geometric distributions. In the cited cases, no network interconnections among elements are considered (e.g., among SDN switches) being their focus on individual nodes (e.g., the controller).

A step forward is made by authors in \cite{mahmood15}, where a Jackson network model is exploited to characterize the interaction between the SDN controller and the switches, which are both modeled as $M/M/1$ systems. Our work, further extend their models, by capturing more sensitive conditions, such as the case of bulk traffic effects. 

Just like us, other authors employ the Jackson network framework. The work in \cite{garzon17} focused on modeling a VNF charaterized by several chained instances. However, they treat a VNF as an individual element, rather than considering it as part of a more complete architecture, which is what we achieve herein.

Open Jackson networks are also used in \cite{zhang17} to model VNF chains in a datacenter. Yet, their focus is on a different problem in relation to optimal VNF placement. 

Finally, authors in \cite{Ye18} consider an $M/D/1$ model to calculate end-to-end packet delay in a flow traversing a node of a VNF-based chain. They present interest findings based on OMNet++ simulations, but do not consider additional metrics as we do herein. 

In another track of works, a more explicit attention is paid to characterize the IMS framework by means of queueing theory models. Authors in \cite{amooee09} and \cite{mishra14} present valuable analyses of delay and bandwidth utilization, respectively. Both works focus on the features of single servers, without considering, as we do, the distinguished chain structure of IMS.

Interesting is also the work in \cite{chi08}, where a queueing model is presented to characterize the behavior of \textit{Notify} messages across an IMS presence server, starting from an analysis of the traffic load distribution. Also in this case, the analysis is focused on a single element (the presence server) but does not capture the effects produced by other nodes. 


In this work we intend to characterize, as precisely as possible, a containerized IMS service chain, a key element of 5G networks. We can pinpoint a number of novel contributions. First, we statistically model a containerized IMS service chain, exploiting the queueing networks framework to capture the relationships that exist among IMS nodes in terms of queueing features. We also take into account the possibility of bulk requests arrival, deriving a generalized form of Pollaczek-Khinchin formula. Then, we solve a connected optimization problem, which is useful to evaluate the global performance of the cIMS service chain. Finally, we carry out an extensive experimental analysis by exploiting data obtained from a Clearwater platform deployment. 

The following outcomes stem from our analyses:
\begin{itemize}
	\item The Jackson framework fits well the modeling of single class requests (e.g., when all customers belong to a single class) within a chain of nodes, and allows to capture the dynamic behavior of observables (e.g., the mean waiting time) at each single node, where the influence of position within the cIMS chain along with the routing logic emerges;
	\item The mean response time across the whole chain (that is directly connected to SLAs offered by telco operators) is characterized in terms of capacity vectors, namely, a set of weights constituting the constraint of an optimization problem focused on minimizing the total time spent in the system;
	\item BCMP framework is introduced to extend the analysis to multi class job requests and two different comparisons are proposed. The first one against the single class (Jackson) model, whereby it emerges that our model exhibits better results in terms of waiting time, at the cost of a more complex architecture. The second comparison is aimed at evaluating the differences emerging by adopting two different queueing policies across the multi-class setting: FCFS and PS (Processor Sharing).
\end{itemize}

\noindent From a telco provider perspective, the afforded characterization turns to be very useful to capture the insights concerning the mutual influence among the nodes that actually belong to a network chain, such as the considered cIMS infrastructure. As a result, providers can guarantee the offered SLAs by optimizing the trade-off between costs and available resources (in terms of capacity, type of nodes, and admissible configurations).

\section{IMS within a Containerized Environment}
\label{sec:ims}

In this section, it is useful to provide in advance a brief description of the Clearwater architecture which represents the reference framework for our experimental analysis, as described in Section V. This preview is helpful to better understand the relationship between the theoretical approach (queueing networks) and the experimental part (cIMS framework) introduced in this work. 

We highlight that a virtualized and, \textit{a fortiori}, container-based IMS solution can elastically scale out under the control of MANO (MANagement and Orchestration), the layer of the NFV reference architecture \cite{mano} in charge of adding (or removing) resources when required.
Each IMS node is developed as a container, while each container is deployed on a microservice infrastructure. In fact, containers in Clearwater are managed by a container engine (we use Docker in our deployment), which is  installed on a virtual machine. Figure \ref{fig:ims} shows a sketch of the Clearwater architecture. A brief description of the nodes, along with their functionality, is proposed next.

\begin{itemize}
	\item \textit{Bono}: it represents the P-CSCF (Proxy-Call Session Control Function) node that acts as anchor point for clients relying on the the Session Initiation Protocol (SIP). It provides NAT traversal procedures as well. 
	\item \textit{Sprout}: this node implements a SIP router and acts as S-CSCF (Serving) and I-CSCF (Interrogating), simultaneously. The former is in charge of managing SIP registrations, whereas, the latter manages the association between UEs (User Equipments) and a specific S-CSCF. In fact, the Sprout node supports SIP for the communication with P-CSCF, and the Diameter/HTTP protocol to retrieve information from SLF/HSS nodes. 
	\item \textit{Homestead}: this node represents the HSS (Home Subscriber Server) and is involved in the users authentication procedures.  
	\item \textit{Ralf}: it acts as a CTF (Charging Trigger Function) module, and is involved in charging and billing operations.
	\item \textit{Homer}: this node manages the service setting documents per user, by acting as an XML Document Management Server (XDMS).
\end{itemize}

\noindent It is useful to underline that, in this work, we model all essential (and mandatory) nodes (the ones enclosed in a red dashed rectangle in Fig. \ref{fig:ims}) which are needed to implement a working IMS, namely: P-CSCF, S/I-CSCF, HSS.

\begin{figure}[t]
	\centering
	\captionsetup{justification=centering}
	\includegraphics[scale=0.31,angle=90]{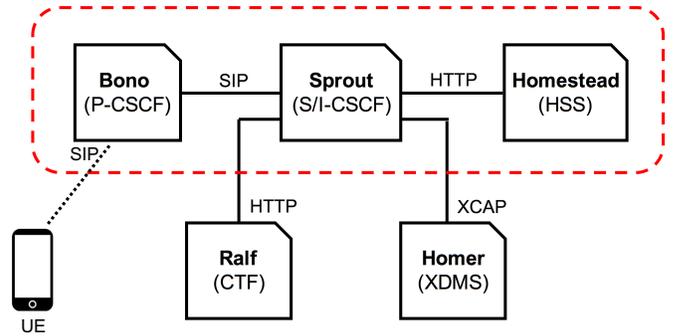}
	\caption{Sketch of Clearwater IMS architecture.}
	\label{fig:ims}
\end{figure} 

\section{The Queueing Networks Model}
\label{sec:qmodel}

In this section, we introduce some details about the queueing networks methodology that we adopt to model the cIMS infrastructure. It is worthwhile recalling that the queueing networks framework is particularly suited to tackle the case of multiple nodes arranged in chains (as it occurs in the considered cIMS scenario), whereby the interconnections among nodes influence the queues distributions. Indeed, a delay caused by an increasing-size queue at a node, affects all the operations that will be performed at the downstream nodes, according to a cascade effect. 

For the sake of simplicity, we start by recasting the interconnection scheme of Fig. \ref{fig:ims} in the model of Fig. \ref{fig:ims_queue}. During this operation, and aimed at considering an even more realistic scenario, we introduce the SLF (Subscriber Location Function) node that routes requests with probabilities p$_1$, p$_2$, and p$_3$ to nodes HSS$_1$, HSS$_2$, and HSS$_3$, respectively, associated to three kinds of user profiles. In practical IMS deployments, telecom operators differentiate their SLAs by means of multiple HSSs governed by an SLF, which is in charge of forwarding requests among HSSs. 
At this stage, it is useful to clarify that the following analysis is split in two: on one hand, we consider a ``regular" case dealing with the standard functioning of the IMS system, whereby each request is processed in a chained way by the series of network nodes, and where classic network queueing theory fits well. On the other hand, we consider a ``special" case, taking into account the problem of requests arriving in bulk, representing events that can occur occasionally  (typically in conjunction with elections, important sporting events etc.). For convenience, we start by presenting this latter case.

\begin{figure}[t]
	\centering
	\captionsetup{justification=centering}
	\includegraphics[scale=0.3,angle=90]{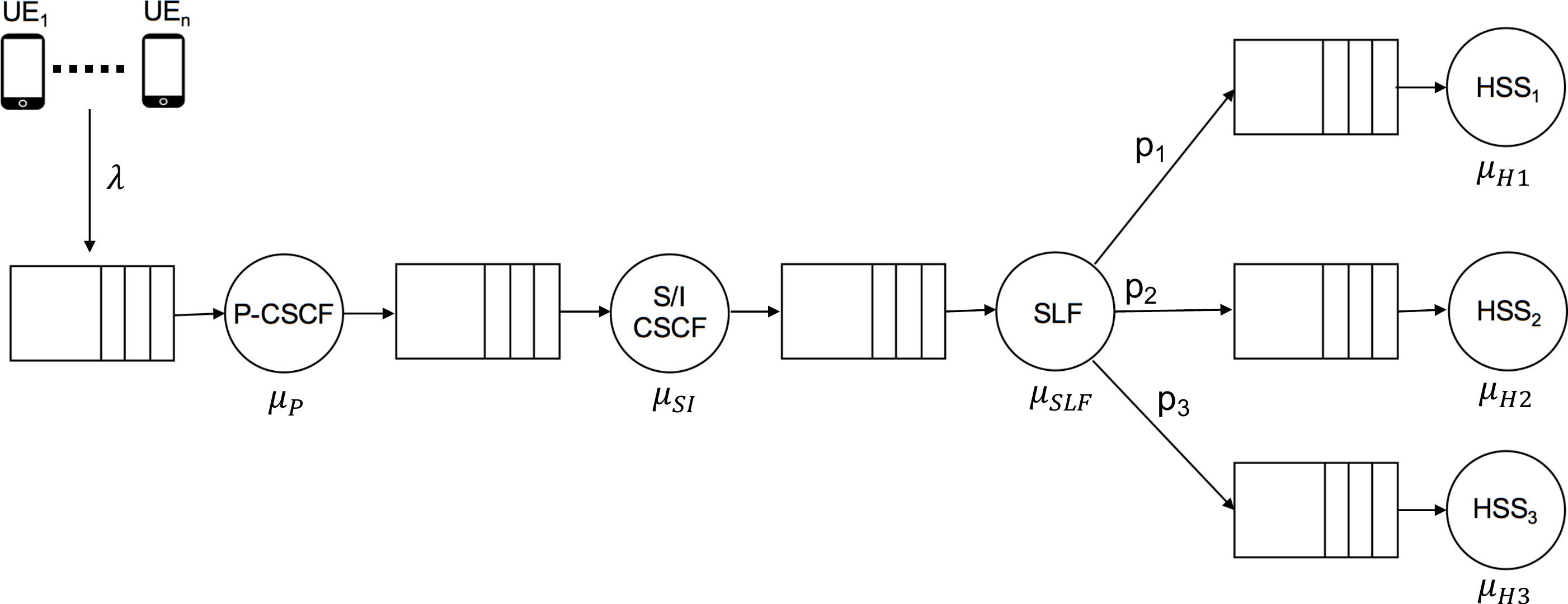}
	\caption{Containerized IMS queueing networks model.}
	\label{fig:ims_queue}
\end{figure} 

\subsection{Bulk arrivals case}

In this section, we consider the P-CSCF node to deal with the special case of bulk arrivals. We want to remark that the functionality of managing bulk requests can also be delegated to a dedicated upstream node (eventually, a load balancer) in charge of selecting more than one softwarized IMS chain to process the requests. In order to address this particular case, we consider an $M/G/1$ queue (requests arrive according to a Poisson process whereas service times have a generic distribution), which allows us to arrive at an extended version of the so-called Pollaczek-Khinchin (P-K) formula that, in classic literature (\cite{bgbook,kleinbook}), is typically derived with no reference to the bulk case.

Let us define some useful quantities: $A(t)$ is the number of requests which arrive at node in the interval $[0,t]$; $A_b(t)$ is the number of bulks of requests which arrive at node in the interval $[0,t]$; given $b_k$ the size of $k$-th bulk, we also have:
\beq
A(t)=\sum_{k=1}^{A_b(t)}b_k.
\label{eq:bulk}
\eeq 
Moreover, the mean bulk arrival rate $\lambda_b$ and the mean (overall) arrival rate $\lambda$ can be defined, respectively, as
\beq
\lambda_b=\lim\limits_{t \rightarrow \infty} \frac{A_b(t)}{t}, ~~~ \lambda=\lim\limits_{t \rightarrow \infty} \frac{A(t)}{t}.  
\label{eq:lambdas}
\eeq
The relationships between $\lambda_b$ and $\lambda$ defined in (\ref{eq:lambdas}) can be derived through the following Proposition.

\begin{lemma}
	By assuming that $\lambda_b$ and $\mathbb{E}[b]$ (the average bulk size) exist and are finite, we have:
	$\lambda=\lambda_b \mathbb{E}[b]$.
\end{lemma}
\begin{proof}
	Starting by definition in (\ref{eq:lambdas}) one has:
	\beqa
	\lambda&=&\lim\limits_{t \rightarrow \infty}\frac{A(t)}{t} \\ \nonumber
	&=& \lim\limits_{t \rightarrow \infty} \frac{1}{t} \sum_{k=1}^{A_b(t)}b_k \\ \nonumber
	&=& \lim\limits_{t \rightarrow \infty} \frac{A_b(t)}{t} \frac{1}{A_b(t)} \sum_{k=1}^{A_b(t)}b_k =  \lambda_b \mathbb{E}[b].
	\eeqa
\end{proof}

Indicating by $\mathbb{E}[S]$ the mean service time of the node, the utilization factor $\rho$, namely, the proportion of time during which the node is busy, can be accordingly defined as $\rho=\lambda \mathbb{E}[S] = \lambda_b \mathbb{E}[b] \mathbb{E}[S]$, where the stability condition $\rho<1$ holds.    
It is now interesting to derive an expression for the mean waiting time at the entry of the P-CSCF node, provided that requests arrive often in bulks.
We start from a known procedure (see \cite{bgbook}) that allows to derive the P-K formula for a $M/G/1$ system queue. 
Suppose that service times are represented by i.i.d. random variables $S=(S_1, \dots, S_s)$. The P-K formula provides an expression for the expected request waiting time in queue $W$, and admits the following expression:

\beq
\mathbb{E}[W]=
\frac{\lambda \mathbb{E}[S^2]}{2(1-\rho)},
\label{eq:pk}
\eeq
where $\mathbb{E}[S^2]$ is the second moment of service time. In case of $M/M/1$ system $\mathbb{E}[S^2]=2/\mu^2$, and, the equation (\ref{eq:pk}) becomes
\beq
\mathbb{E}[W]=\frac{\rho}{\mu(1-\rho)}.
\label{eq:new_pk}
\eeq
Proof of eq. (\ref{eq:pk}) and, then, (\ref{eq:new_pk}) requires the definition of $R_i$, namely, the residual service time experimented by request $i$ when a prior request is being served (see \cite{bgbook}). 
Defined the mean residual time $R=\lim\limits_{i\rightarrow \infty}E[R_i]$, it is possible to show that $\mathbb{E}[W]=R + \mathbb{E}[S] A_q$, where $A_q$ is the mean number of requests at P-CSCF node that, given the Little's theorem, can be expressed as $A_q=\lambda \mathbb{E}[W]$. 
Thus, by a trivial substitution we finally get\footnote{Such a formula can be found in \cite{bgbook} - eq. (3.47), along with the proof of residual time derivation. For the proof of the version with bulk requests (not afforded in \cite{bgbook}) we maintain a coherent notation.}:   
\beq
\mathbb{E}[W]=R+ \mathbb{E}[S] A_q = R+ \lambda \mathbb{E}[S] \mathbb{E}[W] = R+ \rho\mathbb{E}[W],
\label{eq:residual}
\eeq
where $R=\lambda \mathbb{E}[S^2]/2=\rho/\mu$.
Let now consider the more general case where the requests (in our case IMS registration flows) arrive in bulk, and where the size of bulk $b$ has a certain distribution (and is independent of requests service times). Denoting by $W_b$ the waiting time of a request within a bulk, eq. (\ref{eq:residual}) can be rewritten according the following form: 
\beq
\mathbb{E}[W] = \frac{\rho}{\mu}+ \rho\mathbb{E}[W] + \mathbb{E}[W_b]
\label{eq:totalW}
\eeq
and the following Proposition holds:

\begin{lemma}
	The mean waiting time in queue of an arbitrary request $\mathbb{E}[W_b]$ obeys to:
	\beq
	\mathbb{E}[W_b]=\frac{1}{2 \mu} \left[  \frac{\mathbb{E}[b^2]}{\mathbb{E}[b] }  -1 \right].
	\eeq

\end{lemma}
\begin{proof}
	Let $S_{i,j}$ be the service time (i.i.d.) of request $i$ in the bulk $j$. If $S_n$ is the total waiting time of all requests in a bulk $n$, it is possible to write:
	\beq
	S_n=S_{1,n}+(S_{1,n}+S_{2,n})+\dots+(S_{1,n}+\dots+S_{(Z-1),n})
	\label{eq:sum_wait}
	\eeq
	for $Z\ge2$, and with $S_n=0$ for $Z=0,1$, being $Z$ a random variable representing the bulk size. Moreover, we assume that $S_{i+1,n}>S_{i,n}$ for $i \ge 1$.
	
	We have:
	\beqa
	\label{eq:conditional}
	& & \mathbb{E}[S_n | Z=h]  \\ \nonumber
	&& = \mathbb{E}[S_{1,n}+(S_{1,n}+S_{2,n})+\dots+(S_{1,n}+\dots+S_{h-1,n})] \\ \nonumber
	&& \overset{(a)}{=} \mathbb{E}[S] \frac{h(h-1)}{2}=\frac{1}{\mu} \frac{h(h-1)}{2} ~~~~ (h \ge 0), 
	\eeqa
	where, the equality $\overset{(a)}{=}$ comes from the fact that, considering a stationary queue, the order of requests is irrelevant, thus, the subscripts are suppressed.    
	
	By using (\ref{eq:conditional}), and, posing $ \mathbb{P}(Z=h)=p_h$, we get:
	
	\beqa
	\mathbb{E}[S_n]&=&\sum_{h=1}^{\infty}\mathbb{E}[S_n | Z=h] p_h \\ \nonumber
	&=& \sum_{h=1}^{\infty} \frac{1}{\mu} \frac{h(h-1)}{2} p_h \\ \nonumber
	&=& \frac{1}{2 \mu} \left[ \sum_{h=1}^{\infty} h^2 p_h - \sum_{h=1}^{\infty} h p_h \right] \\ \nonumber
	&=& \frac{1}{2 \mu} \left [\mathbb{E}[b^2] - \mathbb{E}[b] \right]
	\label{eq:essen}
	\eeqa
	thus,
	
	\beq
	\mathbb{E}[W_b]=\frac{\mathbb{E}[S_n]}{\mathbb{E}[b]}=\frac{1}{2 \mu} \left[  \frac{\mathbb{E}[b^2]}{\mathbb{E}[b] }  -1 \right],
	\label{eq:wubi}
	\eeq
	and the Proposition is proved.
	
	Moreover, substituting (\ref{eq:wubi}) in (\ref{eq:totalW}) we obtain:
	\beq
	\label{eq:finale}
	\mathbb{E}[W]= \frac{\rho}{\mu (1-\rho)} + \frac{1}{2 \mu (1-\rho)} \left[  \frac{\mathbb{E}[b^2]}{\mathbb{E}[b] }  -1 \right],
	\eeq
	where the first term of R.H.S. of (\ref{eq:finale}) represents the mean waiting time of requests arriving according to a Poisson process with rate $\lambda$, whereas, the second term indicates the additional mean delay due to bulk arrivals. Obviously, for $b=1$ (corresponding to a single arrival) the second R.H.T. term of (\ref{eq:finale}) vanishes, and we end up again with the classic P-K formula for $M/M/1$ queues. 
\end{proof}
Figure \ref{fig:bulk}, shows the mean waiting time in queue for Poisson arrivals in bulk with a uniform distribution, and with a maximum bulk size amounting to $100$. In fact, being P-CSCF the first contact point of an IMS-based architecture, it can be called to manage bulk traffic by implementing dynamic scaling policies (not faced in this work) allowing to increase computational resources when bulk arrivals occur. 
In the case of Markovian service time assumption, another possibility is to increase the number of instances working in parallel leading to a $M/M/m$ queueing model, so that each request always finds an instance able to serve it, and no bulk is formed. Specifically, the ``regular" case (no exceptional bulk requests) afforded in the next section basically lies on exponential assumptions that we accurately justify in the following.  
	
	\begin{figure}[t]
		\centering
		\captionsetup{justification=centering}
		\includegraphics[scale=0.47]{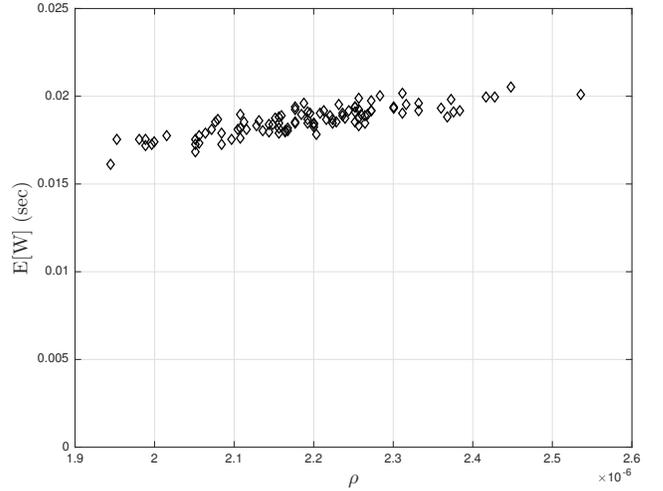}
		\caption{Mean waiting time in queue for Poisson arrivals in bulk with uniform distribution (max bulk size = 100).}
		\label{fig:bulk}
	\end{figure}

\subsection{IMS chain queueing model}
\label{subs:imschain}

Before detailing the network queueing model, we need to clarify some assumptions that allow to reasonably map the theoretical model onto the IMS-based deployment. The first one pertains to the IMS requests arrival times that are supposed to follow a Poisson distribution in accordance to classic teletraffic theory, whereby packets (or calls) originate from a vast population of independent users. This assumption became popular for modeling arrival times in legacy telecommunication networks \cite{poiss_cell}, due to its mathematical tractability. It has subsequently been adopted also in modern data networks when characterizing multimedia traffic. Some examples include: \cite{schu_2000} explicitly focused on exponential arrivals of internet telephony calls; \cite{schu_2008} including the proposal of a SIP simulator where, taking into account also suggestions provided by IETF SIP design team, call generations and call holding times follow an exponential model; \cite{gur_2005} where a SIP proxy server is modeled by means of an M/M/1 queueing system. More recently, authors in \cite{schembra2015} propose a management model for an SDN/NFV customer premises equipment (CPE) node, where the CPE node is supposed to be reached by a Poisson distributed network traffic.
The second assumption involves the Markovian hypothesis about the service times of IMS network nodes. This assumption is justified by the consideration that very long service times occur only occasionally (e.g. when a node is overloaded also by other tasks such us software updates). Whereas, for the remaining time the network node tries to evade the request as fast as possible. Also in this case, scientific literature exhibits valuable examples: for instance, in \cite{Subramanian_2009,Subramanian_2011}, an M/M/1 scheme has been adopted to model a SIP proxy server, where the considered assumptions have been validated in conjunction with CISCO performance team. Based on realistic simulations is also the work of Bell-Labs authors \cite{Chi_2008}, where service processing times (in particular related to SIP PUBLISH messages) are assumed to be exponentially distributed. 

Essentially, an IMS system is nothing but a chained of elements that have to be traversed in a predefined order to provide a specific service (e.g. Registration). This configuration is well suited to be represented by the open Jackson networks formalism. An open network \cite{trivedibook} is a particular type of queueing network where jobs (IMS requests) enter the system from outside according to a Poisson process. Once reached the system (in our case the P-CSCF node), jobs are routed within the chain of nodes and, once service is completed, they leave. This formalism is counterposed to closed networks where the number of jobs entering the system remains constant, since these are being reinserted in the system in a loop fashion. 
In an open network with $N$ nodes, the following balance equation holds:

\beq
\lambda_i=\lambda + \sum_{j=1}^{N}\lambda_j \cdot p_{ji},
\label{eq:balance}
\eeq
where: $\lambda_i$ denotes the overall arrival rate of jobs at the node $i$ ($i=1,\dots,N$), $\lambda$ denotes the arrival rate of jobs from outside\footnote{We consider that external jobs/requests always arrive at P-CSCF before entering the system.}, and $p_{ji}$ denotes the routing probability, namely, the probability that a job is moved to node $i$ once the service at the node $j$ is completed. 
In case that arrivals are Poisson from outside, the service times are exponentially distributed (eventually, each node can be composed of $m_i \geq 1$ service instances), and the service disciplines are FCFS, the system is referred to as an open Jackson network. Again, if in an open network the ergodicity condition $\rho_i<1$ is guaranteed for each node, the steady-state probability of the whole system (network of queues) can be expressed as the product of marginal probabilities of the single nodes: 

\beq
\pi(k_1, k_2, \dots, k_N)=\prod_{i=1}^{N}\pi_i(k_i),
\label{eq:jackson}
\eeq    
where, the joint probability vector on the L.H.S. of (\ref{eq:jackson}) represents the steady-state probability of having $k_i$ jobs at node $i$ ($i=1,2,\dots,N$), whereas, at R.H.S., we have a product of marginal probabilities. Such result (proved in \cite{qnetbook}) is known as the Jackson's Theorem, and the resulting network is often referred to as \textit{product-form} network. 
In the case of $M/M/1$ queues, the marginal probabilities $\pi_i (k_i)$ admit the following  expression:
\beq
\pi_i(k_i)=(1-\rho_i)\rho_i^{k_i},
\label{eq:margprobmm1}
\eeq    
where $\rho_i=\lambda_i/\mu_i$.
In the more general case of $M/M/m$ systems, the marginal probabilities $\pi_i (k_i)$ can be directly derived by  \cite{bgbook}:

\beqa
\centering
\pi_i (k_i)=
\left\{
\begin{array}{l}
{\begin{array}{ll}
	\hspace{-0.2cm} \pi_i(0)\frac{(m_i \rho_i)^{k_i}}{k_i !},  \;\;\;\;\; k_i \le m_i,
	\end{array}}
\\
\\
 \pi_i(0) \frac{m_i^{m_i} \rho_i^{k_i}}{m_i !},  \;\;\;\;\;  k_i > m_i,
\end{array}
\right.
\label{eq:marginal}
\eeqa

where: $\pi_i(0)$ is the steady-state probability, $\rho_i=\lambda_i/m_i \mu_i < 1$ and the condition $\sum_{k_i=0}^{\infty}\pi_i(k_i)=1$ holds. 
When dealing with network queues, another useful parameter to take into account is the mean number of visits $v_i$ of a request at node $i$, defined through the visit ratio (a.k.a. relative arrival rate) $v_i=\lambda_i/\lambda$ which can also be related to routing probabilities by means of the following equation:
\beq
v_i=p_{0i}+\sum_{j=1}^{N}v_j \cdot p_{ji}, 
\label{eq:visitprob}
\eeq
where $p_{0i}$ indicates the probability that a request comes from outside to $i$-th node. Such a measure is helpful to evaluate other quantities such as the mean time spent in the system, that, we characterize in the forthcoming performance assessment. 

\subsection{Optimization Problem}

In practice, many telco providers have to guarantee  SLAs that are often related to time constraints (e.g. delay) which a ``job" has to respect when it enters  a network system. In line with this consideration, let us consider the mean time spent by a job within a generic cIMS node (often called mean response time). This quantity is the sum of time spent in queue and time spent for processing (service time) at each node, and the following equality holds:

\beq
\mathbb{E}[T_i]=\mathbb{E}[W_i]+\mathbb{E}[S_i]=\frac{1}{\mu_i - \lambda_i}, 
\label{eq:tot_wait_time}
\eeq
where, $\mathbb{E}[W_i]$ can be derived by (\ref{eq:new_pk}), whereas, $\mathbb{E}[S_i]=1/\mu_i$ according to the $M/M/1$ assumption. Exploiting the results of the Jackson's Theorem, each single node in the IMS system can be modeled as an $M/M/1$ queue. 
Thus, aimed at minimizing the average total time that a job spends in the cIMS system, we want to solve the following convex optimization problem:

\beqa
	&\textnormal{minimize}& \quad \sum_{i=1}^{N} \frac{1}{c_i \mu_i - \lambda_i}  \nonumber  \\\nonumber\\
	 &\textnormal{subject to}& ~~~ \sum_{i=1}^{N} {c_i \mu_{i}} = C, ~ c_i \mu_i>\lambda_i, ~ \lambda_i \ge 0 
\label{eq:idealopt}
\eeqa
where:
\begin{itemize}
	\item $c_i >0$ is a capacity factor associated to the service rate of a specific node. In real scenarios, this value is related to the computational power (in terms of CPU, RAM, etc.) of a node, which in a cloud environment refers to the possibility of dynamically adjusting virtual resources;
	\item $C>0$ represents the total budget constraint.
\end{itemize}

It is useful to recall that the convenience of convex optimization formulation (when possible) leads to analytical expressions amenable to be solved by means of straightforward calculations. In the considered case, the convexity of problem directly stems from the convexity of function $\sum_{i=1}^{N}\frac{1}{c_i \mu_i -\lambda_i}$ since: $i)$ the term $\frac{1}{c_i \mu_i -\lambda_i}$ admits a positive second derivative with constraint $c_i \mu_i - \lambda_i >0$; $ii)$ the overall summation is again a convex function since it is a linear combination of convex functions with non-negative coefficients.  

Now, given a Lagrange multiplier $\mathscr{L}$, the optimization problem in (\ref{eq:idealopt}) can be rewritten as dual form:

\beqa
	&\textnormal{minimize}& \quad \sum_{i=1}^{N} \frac{1}{c_i \mu_i - \lambda_i} +  \mathscr{L}  \sum_{i=1}^{N}{c_i \mu_i}  \nonumber  \\\nonumber\\
	&\textnormal{subject to}& \quad  \mathscr{L} > 0, ~ c_i \mu_i>\lambda_i,  ~ \lambda_i \ge 0 .
\label{eq:lagrange}
\eeqa
It is possible to separately optimize the variables $\mu_i$ in problem (\ref{eq:lagrange}); thus, we have to find the optimal $\mu_o$ that minimizes the following Lagrangian:
\beq
\beta(\bm{\mu})=\frac{1}{c_o \mu_o - \lambda_o} + \mathscr{L} c_o \mu_o.
\eeq
The optimal solutions are obtained by nullifying the partial derivatives:

\beqa
\frac{\partial \beta}{\partial \mu_o}&=&- \frac{c_o}{(c_o \mu_o - \lambda_o)^2} + \mathscr{L} c_o =0\nonumber\\
&\Rightarrow&  \mu_o= \frac{1}{c_o} \left( \lambda_o + \frac{1}{\sqrt{\mathscr{L}}}  \right).
\label{eq:partialderiv}
\eeqa
By imposing the constraint in (\ref{eq:idealopt}), we can write: 

\beqa
\sum_{i=1}^{N} c_i \mu_i = C =  \sum_{i=1}^{N} \left(\lambda_i + \frac{1}{\sqrt{\mathscr{L}}} \right), 
\eeqa
that, after straightforward algebraic manipulations, leads to:
\beq
\frac{1}{ \sqrt{\mathscr{L}}} = \frac{C - \sum_{i=1}^{N} \lambda_i}{N}.
\label{eq:optlagr}
\eeq
Substituting (\ref{eq:optlagr}) in (\ref{eq:partialderiv}) we get the desired solution:
\beq
\mu_o = \frac{\lambda_o}{c_o} + \frac{C - \sum_{i=1}^{N} \lambda_i}{c_o N}.
\label{eq:muopt}
\eeq
This result can be interpreted as a variant of the optimal capacity allocation problem, as originally formulated by Kleinrock \cite{kleinbook}, and admits the following interpretation: the first term on R.H.S. of (\ref{eq:muopt}) accounts for the capacity allocation assigned to each node aimed at satisfying effective arrival rates; whereas, the second term accounts for an extra capacity distributed among other nodes. As the total number of nodes grows asymptotically ($N \rightarrow \infty$), it is possible to neglect the second term, thus, only the effective capacity assigned to a specific node is considered.
In the end, the optimal assignment of capacity factors (guaranteed by the solution of the analyzed convex optimization problem) can also be interpreted as the optimal allocation (or tuning) of additional instances \textit{m}, which a service provider can activate to counter a given mean response time constraint.

\section{Performance Assessment}
\label{sec:exp}

We start by arranging from scratch an experimental testbed of a cIMS infrastructure which will allow to collect realistic, experimental data (e.g. service times of cIMS nodes) that will, in turn, be useful to calculate metrics of interest (e.g. mean queue length, mean waiting time, etc.). Then, we carry out a performance evaluation that can be split in two parts: the first one is aimed at assessing the performance of a scenario where cIMS requests belong to the same class (Single Class Analysis), along with the evaluation of the optimal cIMS deployment w.r.t. a capacity constraint. In the second part, we extend the assessment to the case of cIMS requests differentiated per class (Multi Class Analysis), where we also consider the case of different queueing strategies.
In practice, such comparative analysis accounts for two models relying on the same intuition of characterizing a chained system in terms of the intermediate nodes queueing behavior: Jackson networks (previously described), useful to afford the Single Class Analysis, and BCMP networks amenable to tackle the Multi Class Analysis.

\subsection{Experimental setting}
We now provide some useful details about the developed testbed relying on a Clearwater architecture deployment. The architecture considered for our setting (mainly inspired to a similar deployment in \cite{cotroneo17bis}) consists of a hosting machine equipped with an Intel Xeon $4$-core $3.70$GHz, $32$ GB of RAM and a VMware-based hypervisor. 
We deploy three different VMs each of which hosts on top the containerized functionalities: P-CSCF (\textit{Bono}), S/I-CSCF \textit{(Sprout)}, HSS (\textit{Homestead} including \textit{Cassandra} DB for storing users information and profiles). Each VM is equipped with a (virtual) 2-Core CPU and 8 GB of RAM. A test VM based on a Linux distribution (mounted on a separate hardware) and connected via Gigabit Ethernet LAN acts as a stress node equipped with SIPp, an opensource tool amenable to be scripted for simulating workload.

The performed tests allowed us to simulate the initialization of $1000$ IMS sessions with a BHCA (Busy Hour Call Attempts) equal to $2.6$ per user (in line with values provided for VoLTE -  see \cite{tonse}). As a result, we derive an estimate of Registration Delay (RD), defined as the time interval between a Register message (originated from a caller UE) to the 200 OK message (sent back to caller from S-CSCF node when procedure ends correctly). This mean value amounts to about $30$ msec and is in line with standard RD values (see \cite{ltebook}). On the other hand, we carried out a more detailed analysis on a sample of $10$ IMS Register sessions (by means of network sniffer Wireshark) aimed at retrieving the mean time that each cIMS node spends in processing a request. This value can be interpreted as the mean service time ($1/\mu$) per node and is in the order of few milliseconds for each node. 
Table \ref{tab:params} summarizes the input parameters that we derive from the experimental analysis, whereby, for the case of SLF node, we consider values in line with its forwarding activity. In the case of routing probabilities ($p_1$, $p_2$, $p_3$), instead, we merely consider exemplary values that can be obviously tuned according to specific deployments.

\subsection{Single Class Analysis (Jackson framework)}

\begin{figure*}
	\centering
	\begin{subfigure}[t]{0.45\textwidth}
		\centering
		\includegraphics[width=7.6cm]{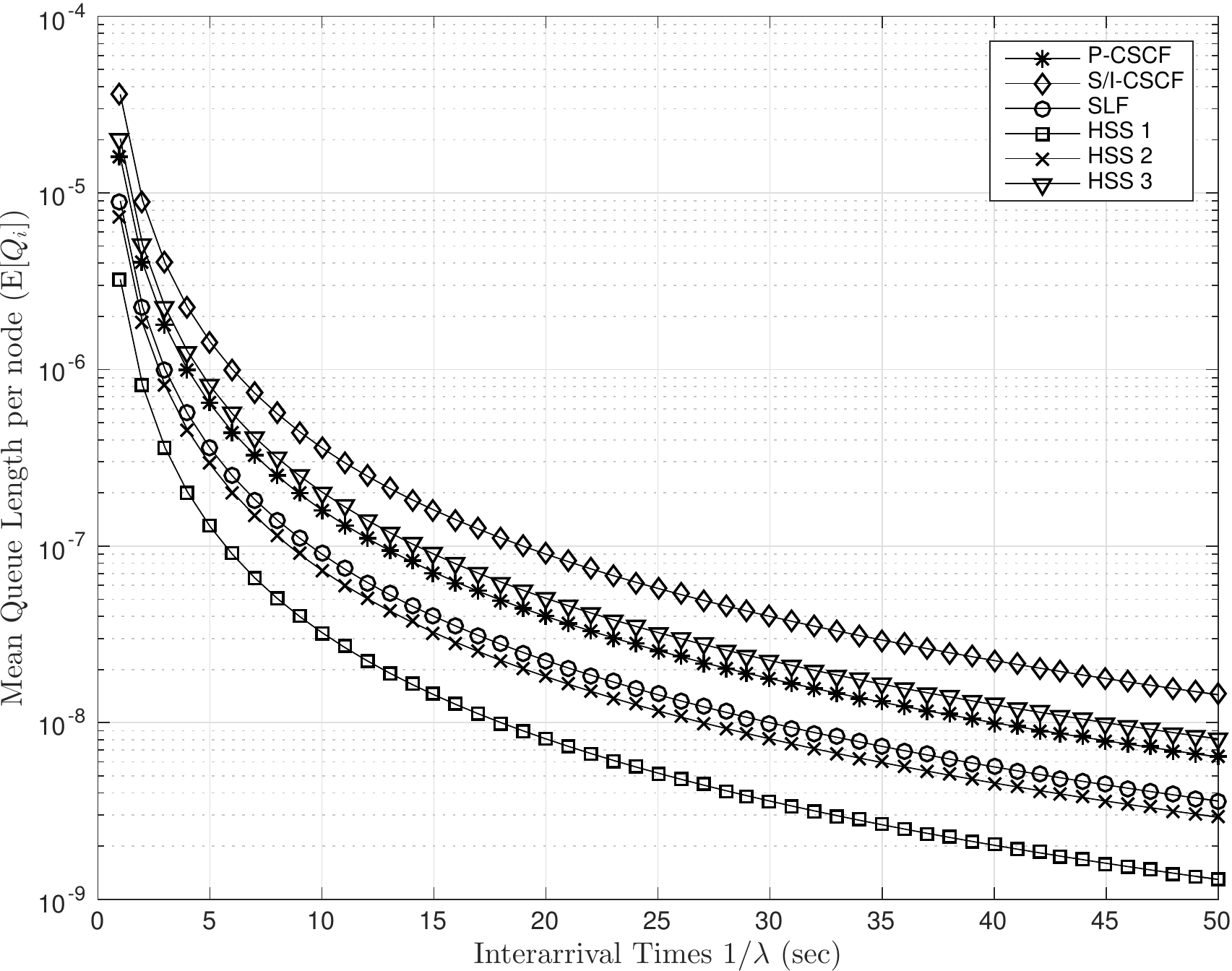}
		\caption{}
	\end{subfigure}
	\vspace{6.5mm}
	\hspace{15mm}
	\begin{subfigure}[t]{0.45\textwidth}
		\centering
		\includegraphics[width=7.6cm]{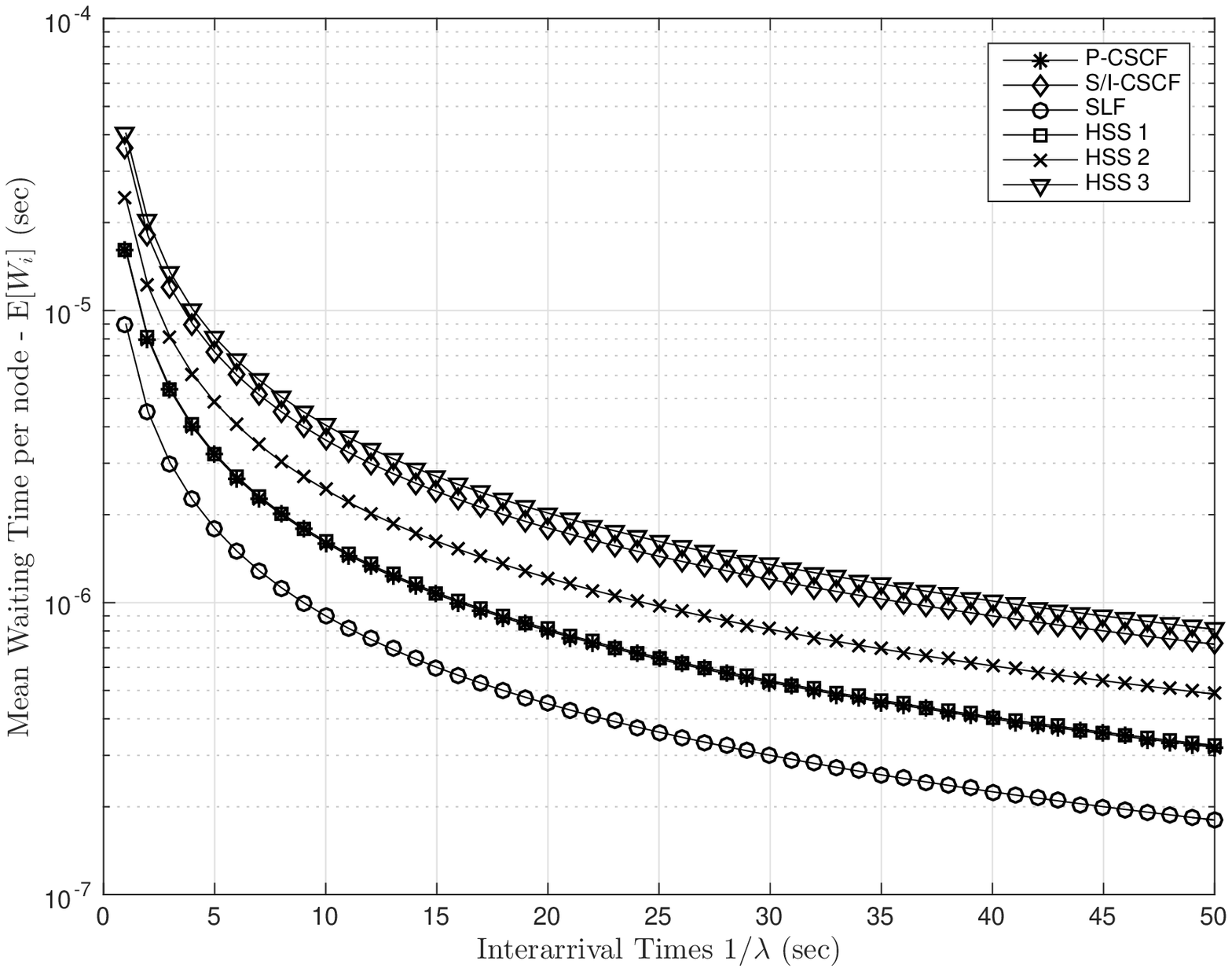}
		\caption{}
	\end{subfigure}
	\begin{subfigure}[t]{0.45\textwidth}
		\centering
		\includegraphics[width=7.6cm]{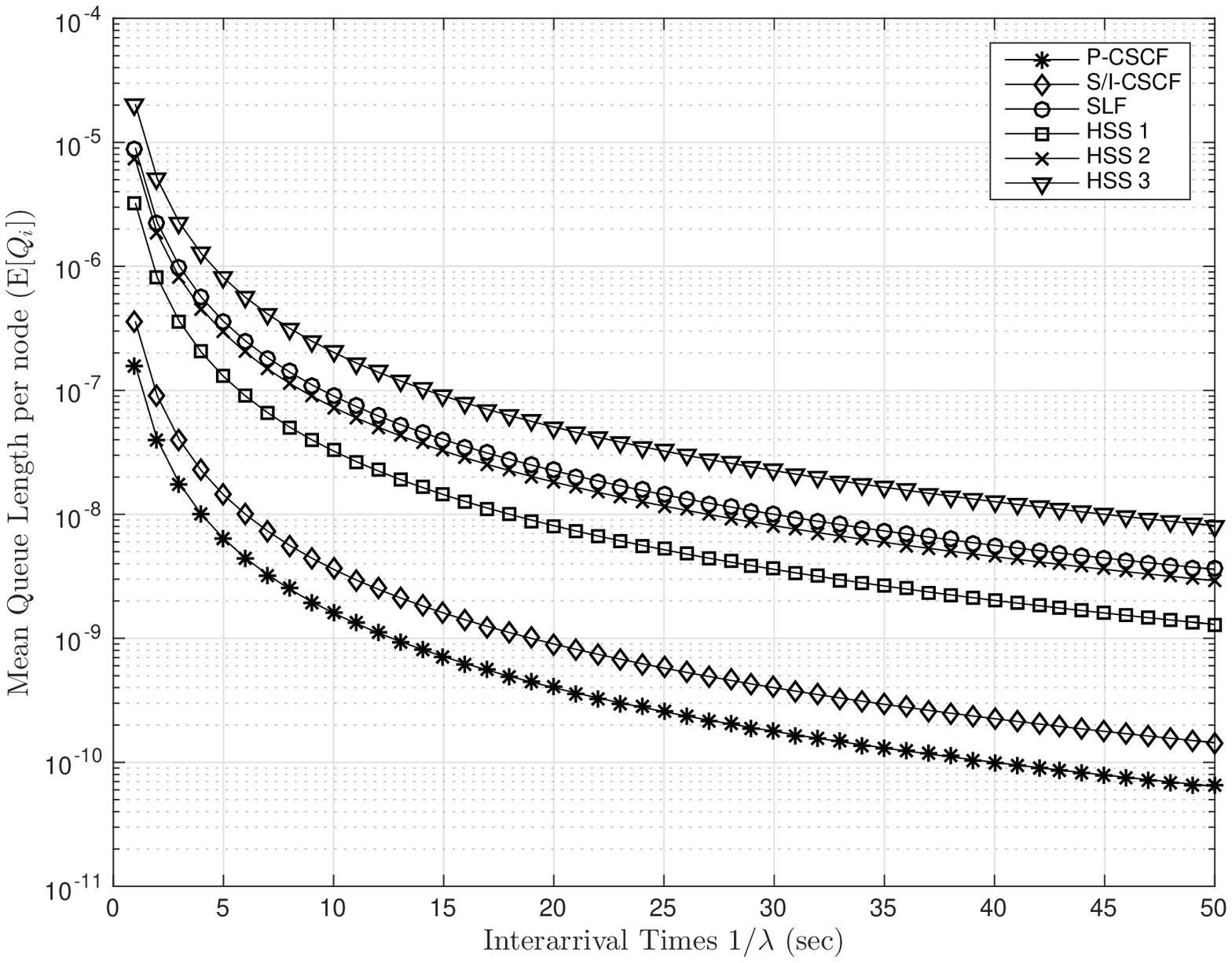}
		\caption{}
	\end{subfigure}
	\vspace{6.5mm}
	\hspace{15mm}
	\begin{subfigure}[t]{0.45\textwidth}
		\centering
		\includegraphics[width=7.6cm]{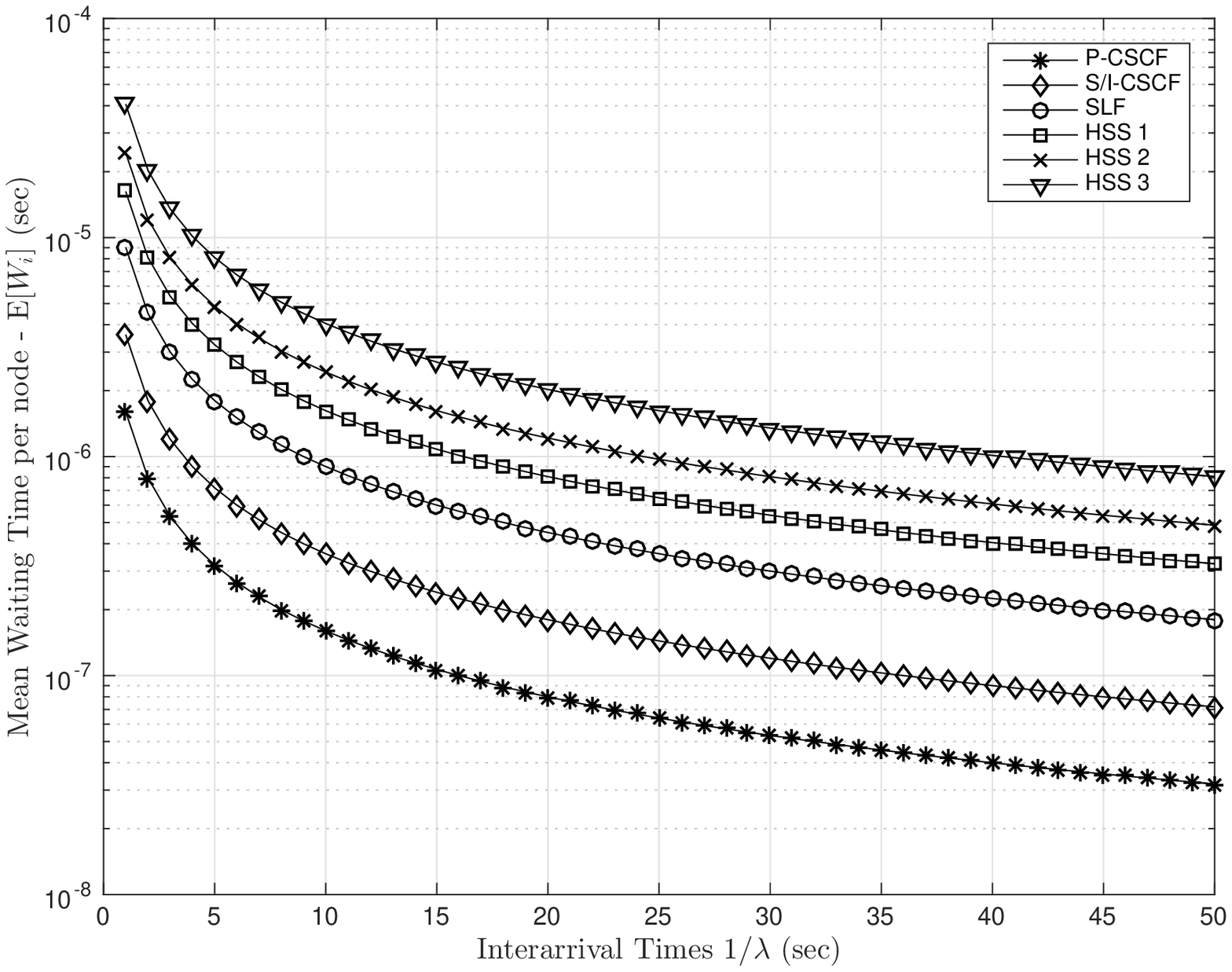}
		\caption{}
	\end{subfigure}
	\caption{Single class analysis. Mean Queue Length (a) and Mean Waiting Time (b) per node ($M/M/1$ model). Mean Queue Lenght (c) and Mean Waiting Time (d) per node ($M/M/1$ model per node excepting for P-CSCF and S-CSCF adopting $M/M/m$ model with $m$=10). }
	\label{fig:singleclassanal}
\end{figure*}

In this scenario, we consider the case of cIMS requests belonging to the same class by exploiting the properties of Jackson's theorem introduced in Section \ref{subs:imschain}.
Let us start analyzing the behavior of cIMS nodes arranged in a network queue fashion where a single class of requests is permitted. Simulations have been realized using the Qnetwork package \cite{marzolla} that allows representing the nodes interconnections by means of the routing matrix. The mean queue length $\mathbb{E}[Q_i]$ per node (accumulated across all visits) where external Poisson requests with rate $\lambda$ occur, can be expressed as

\beq
\mathbb{E}[Q_i]=\frac{\rho_i^2}{1-\rho_i},~~~ \rho_i=\frac{\lambda_i}{\mu_i},
\label{eq:mql}
\eeq
where intermediate arrival rates $\lambda_i$ can be derived from (\ref{eq:balance}).
As a general trend, Figure \ref{fig:singleclassanal}(a) reveals that, as inter-arrival times grow (corresponding in decreasing arrival rates), the mean queue length per node diminishes, as it was to be expected. Now, if we focus on specific nodes, from (\ref{eq:mql}) we can deduce that, for a fixed $\lambda_i$, $\mathbb{E}[Q_i]$ decreases as the service rate of $i$-th node increases. It is interesting to notice that this behavior seems to be violated by the three HSSs (in particular by HSS$_1$ and HSS$_2$) since they exhibit the lowest service rate (or the highest service time, according to the parameters provided in Table \ref{tab:params}). This phenomenon clearly depends on the routing probabilities that, according to (\ref{eq:balance}), act as weights for $\lambda_i$ terms and produce the global effect of reducing the mean queue length for HSS nodes. 

Let us now consider the mean waiting time per node $\mathbb{E}[W_i]$ (accumulated across all visits) that, by applying Little's theorem to (\ref{eq:mql}), can be expressed as

\beq
\mathbb{E}[W_i]= \frac{1}{\lambda_i} \mathbb{E}[Q_i] =  \frac{\rho_i}{\mu_i(1-\rho_i)},~~~ \rho_i=\frac{\lambda_i}{\mu_i}.
\label{eq:mwt}
\eeq

\begin{table}[t!]
	\caption {Input parameters} \label{tab:params}
	\resizebox{.48\textwidth}{!}{
		\begin{tabular}{c|c|c}
			\hline
			Parameter & Description & Value\\
			\hline
			\\[-8pt]
			$1/\lambda$ & outside arrival times & [1 50] sec \\ \hline
			$1/\mu_{P}$ & P-CSCF mean service time & 4$\cdot10^{-3}$ sec \\ \hline
			$1/\mu_{SI}$ & S/I-CSCF mean service time & 6$\cdot10^{-3}$ sec \\ \hline
			$1/\mu_{SLF}$ & SLF mean service time & 3$\cdot10^{-3}$ sec \\ \hline
			$1/\mu_{HSS_i}$ & HSS$_i$ mean service time ($i=1, 2, 3$) & 9$\cdot10^{-3}$ sec\\ \hline
			$p_1$ & routing probability to HSS$_1$ & 0.2  \\ \hline 
			$p_2$ & routing probability to HSS$_2$ & 0.3  \\ \hline 
			$p_3$ & routing probability to HSS$_3$ & 0.5  \\ \hline 
			\hline
		\end{tabular}}
	\end{table}

Figure \ref{fig:singleclassanal}(b) shows the mean waiting time per node. Also in this case the general trend is expected since, as inter-arrival times grow, the mean waiting time per node decreases. In other words, when arrival rates decrease, requests spend less time to be served in a node. As can be argued by (\ref{eq:mwt}), the behavior is similar to the one exhibited for $\mathbb{E}[Q_i]$, with the difference that the service time per node acts as a weight factor. As a result, the curves pertinent to HSS$_1$ and HSS$_2$ tend to grow due to the service time value.

In practice, when dealing with the container technology it is easy to replicate a software instance (e.g. a container functionality) with the aim of exploiting parallel resources. This case can be quickly embodied in the proposed queueing networks framework by admitting that nodes can be modeled as $M/M/m$ queues (remaining in the Jackson's theorem hypotheses) where $m$ represents the number of instances working in parallel, and where $\rho_i=\lambda_i/m_i \mu_i$. Let us assume to model only P-CSCF and S/I-CSCF in terms of $M/M/m$ queues. Figures \ref{fig:singleclassanal}(c) and \ref{fig:singleclassanal}(d) show, respectively, mean queue length and mean waiting time per node, when P-CSCF and S/I-CSCF are modeled as $M/M/10$ queues. For both cases, the overall effect is an expected downward curve scaling for P-CSCF and S/I-CSCF nodes, due to the scaling factor in the $\rho$ expression.   

Let us now focus on the mean response time of the overall cIMS system $\mathbb{E}[T]$, whose single contributions per nodes obey to (\ref{eq:tot_wait_time}). Figure \ref{fig:resp_time} shows the behavior of $\mathbb{E}[T]$ for different values of capacity factors introduced in the previous section. For the sake of simplicity, we denote by $\textbf{c}=[P, S, SLF, H_1, H_2, H_3]$ the vector of capacity factors associated to P-CSCF, S/I-CSCF, SLF, HSS$_i$ (i=1, 2, 3) nodes, respectively. The uppermost curve (denoted by triangular markers) represents a reference case since capacity factors amount to $1$ for each node. This means that nodes work at their nominal conditions with no extra ``power" added. The remaining three curves refer to different cases of capacity factors all summing to $18$, but differently distributed among nodes. For instance, when assigning more power to HSSs ($\textbf{c}=[1, 1, 1, 6 , 5, 4]$), $\mathbb{E}[T]$ decreases from a regime value\footnote{Regime value is intended as a value reached when $1/\lambda$ grows enough to produce negligible variations of $\mathbb{E}[T]$.} of about $22$ msec to about $15$ msec (curve with asterisk markers). This value further diminishes when capacity is differently allocated, by assigning extra power to P-CSCF, S/-CSCF, and SLF nodes, and by leaving HSSs to their nominal value (case $\textbf{c}=[6, 5, 4, 1 , 1, 1]$ and curve with diamond markers). Here, it is interesting to observe that this behavior comes from the fact that HSSs work at a nominal service time higher than one exhibited by remaining nodes. Thus, capacity factors have more effect when applied to P-CSCF, S/I-CSCF, and SLF. Finally, when the power is equally distributed among all nodes (case $\textbf{c}=[3, 3, 3, 3 , 3, 3]$ and curve with square markers), $\mathbb{E}[T]$ decreases below $8$ msec. Accordingly, the latter appears to be the more advantageous configuration (at the same capacity vectors) in case a provider would guarantee Service Level Agreements based on minimum response time of the system by having a fixed cost constraint. 

It is worth remarking that, for comparison purposes, all the curves have been represented on the same plot but, due to different scales, they appear to be flattened around the pertinent regime value. As a matter of fact, we propose a zoom of a part of the transient region ($1/\lambda \in [1,20]$) corresponding to the reference case (see inset pointed by red arrow), where it is possible to appreciate the correct decay of  $\mathbb{E}[T]$, as arrival times increase.    

The results obtained in Fig. \ref{fig:resp_time} can be also verified by means of an asymptotic bounds analysis, which is useful to derive upper and lower bounds for system throughput and mean response time, respectively \cite{den78}. 

\begin{figure}[t!]
	\centering
	\captionsetup{justification=centering}
	\includegraphics[scale=0.37,angle=90]{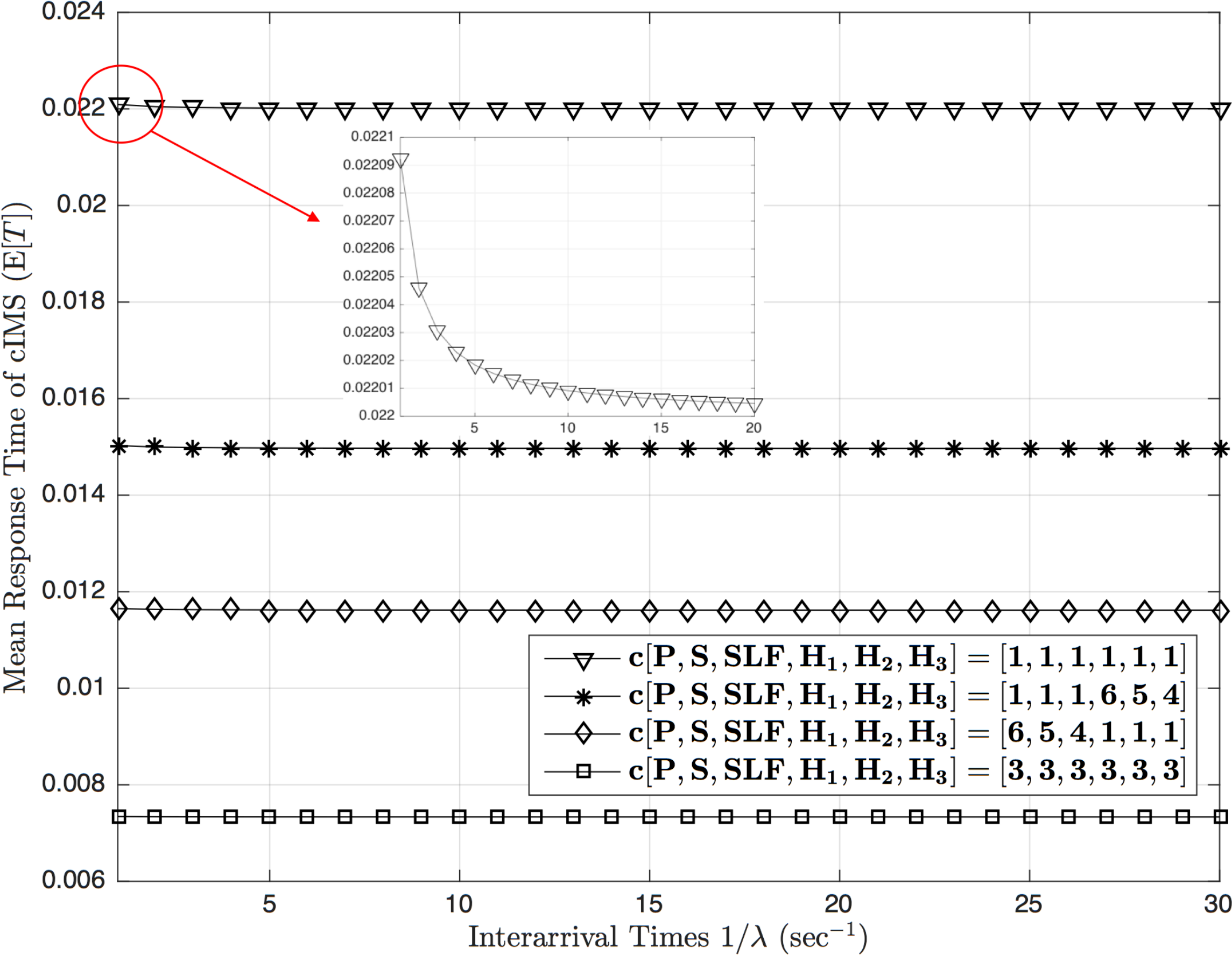}
	\caption{Mean Response Time of the overall cIMS chain for different capacity vectors.}
	\label{fig:resp_time}
\end{figure} 

Having been satisfied the needed condition for this analysis, namely that service rates must be independent of number of requests (at a node or in che cIMS), we define the relative utilization of node $i$ as the quantity $u_i=v_i/\mu_i$. With the assumption that waiting time of a request is zero (best case when there is no request blocked by other requests), and being $u_i$ the mean time a request spend being served at $i-$th node, the mean system response time is given by the sum of relative utilizations. Consequently, the lower (optimistic) bound on mean response time can be expressed as

\beq
\mathbb{E}[T] \geq \sum_{i=1}^{N} \frac{v_i}{\mu_i}.
\label{eq:bound}
\eeq
For the reference case, such bound amounts to $\mathbb{E}[T]=0.022$ sec that, as can be easily verified by inspecting the zoomed section in Fig. \ref{fig:resp_time}, corresponds to limiting value as the interarrival times grow. 
	
 \begin{figure}[t!]
 	\centering
 	\captionsetup{justification=centering}
 	\includegraphics[scale=0.32,angle=90]{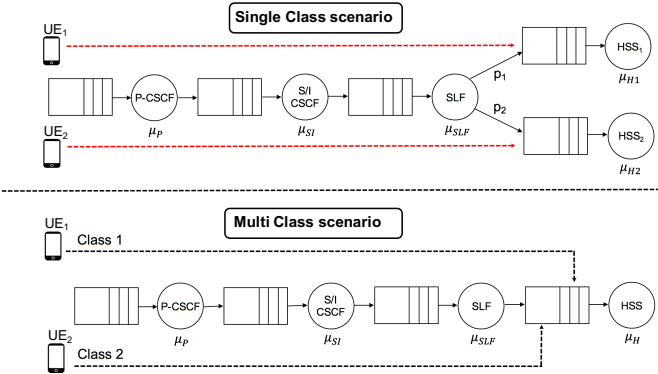}
 	\caption{Single Class scenario (uppermost panel) vs. Multi Class scenario (lowermost panel).}
 	\label{fig:multiclass}
 \end{figure} 

  \begin{figure*}
  	\centering
  	\begin{subfigure}[t]{0.4\textwidth}
  		\centering
  		\includegraphics[width=7.5cm]{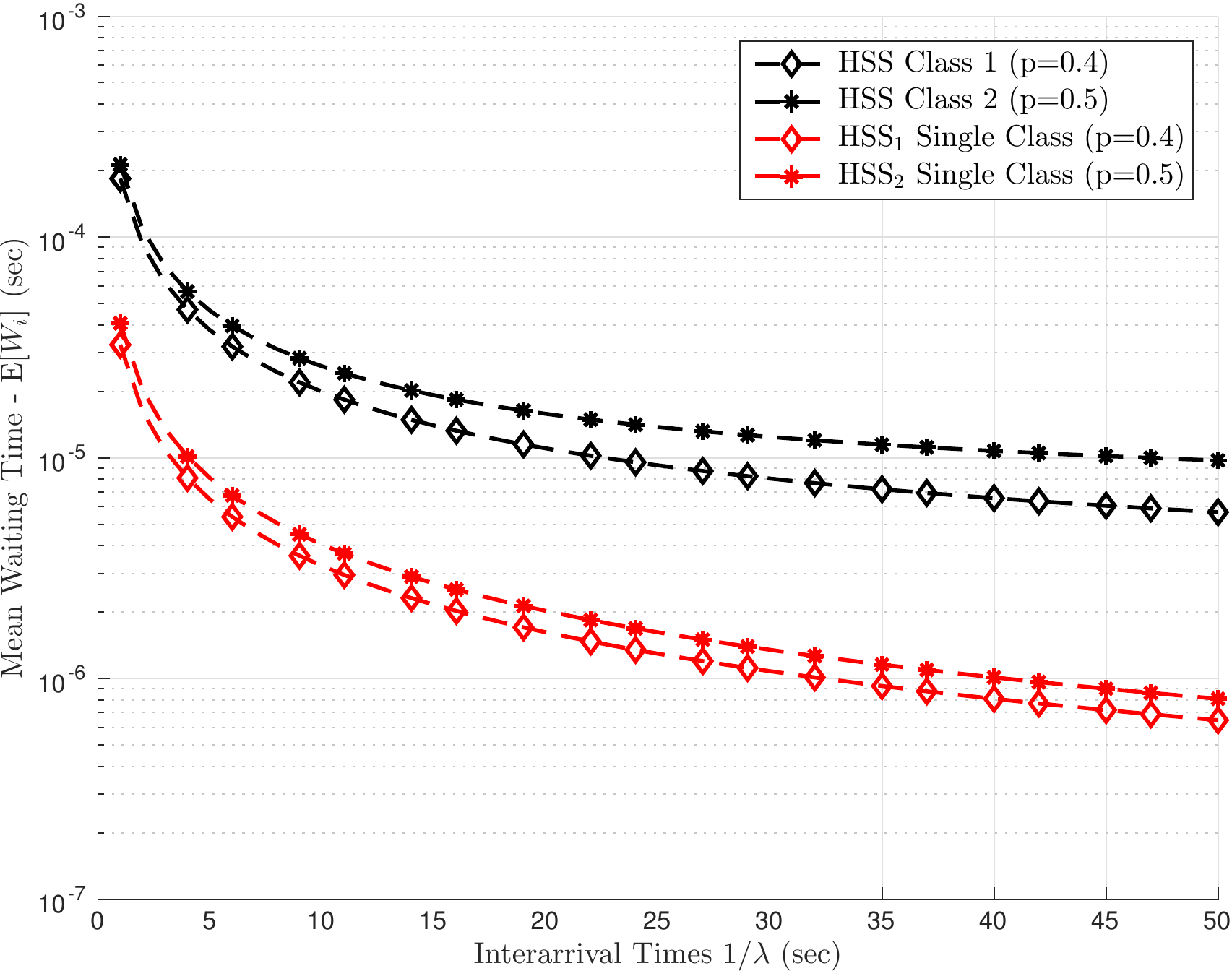}
  		\caption{}
  	\end{subfigure}
  	\vspace{6.5mm}
  	\hspace{15mm}
  	\begin{subfigure}[t]{0.4\textwidth}
  		\centering
  		\includegraphics[width=7.5cm]{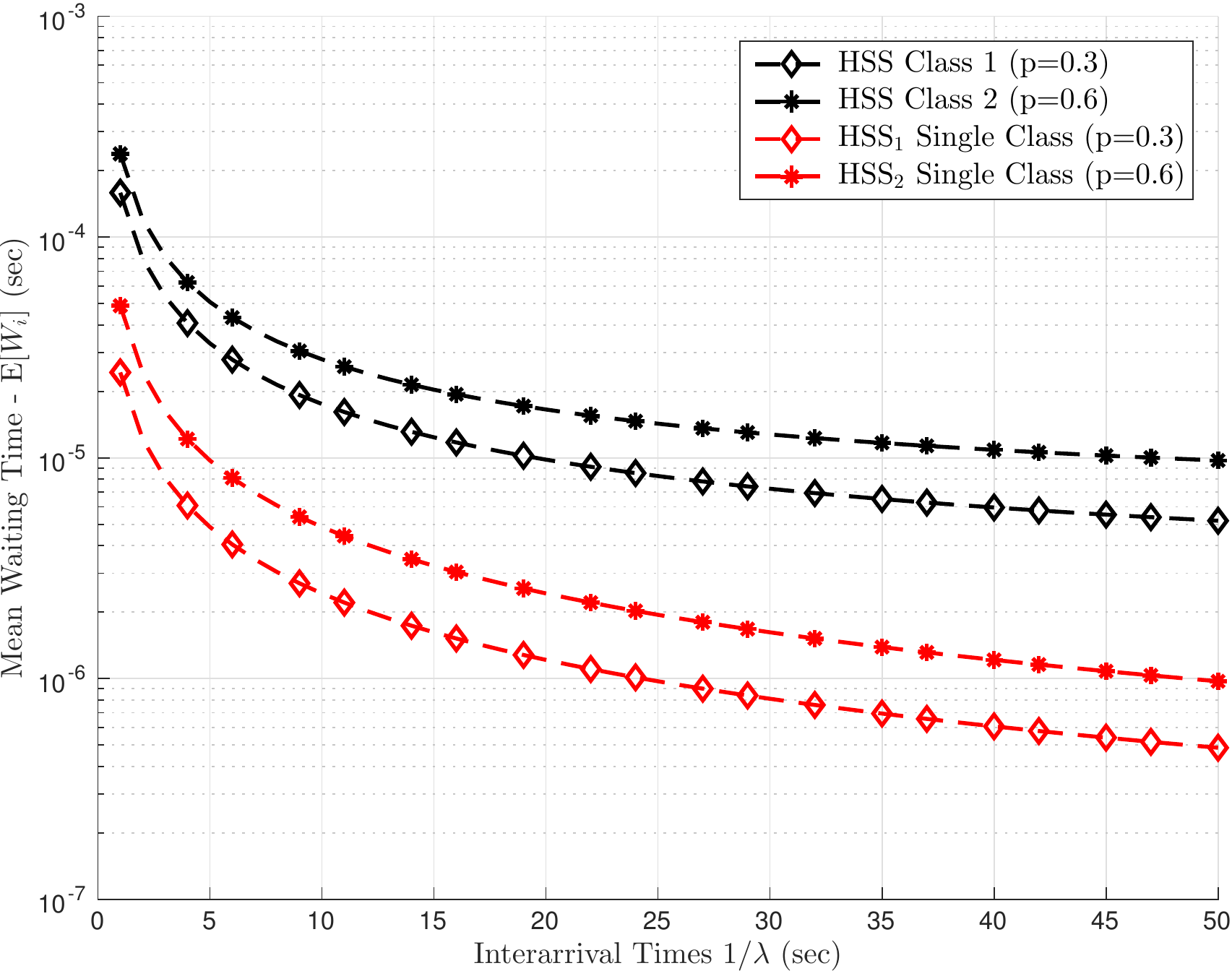}
  		\caption{}
  	\end{subfigure}
  	\begin{subfigure}[t]{0.4\textwidth}
  		\centering
  		\includegraphics[width=7.5cm]{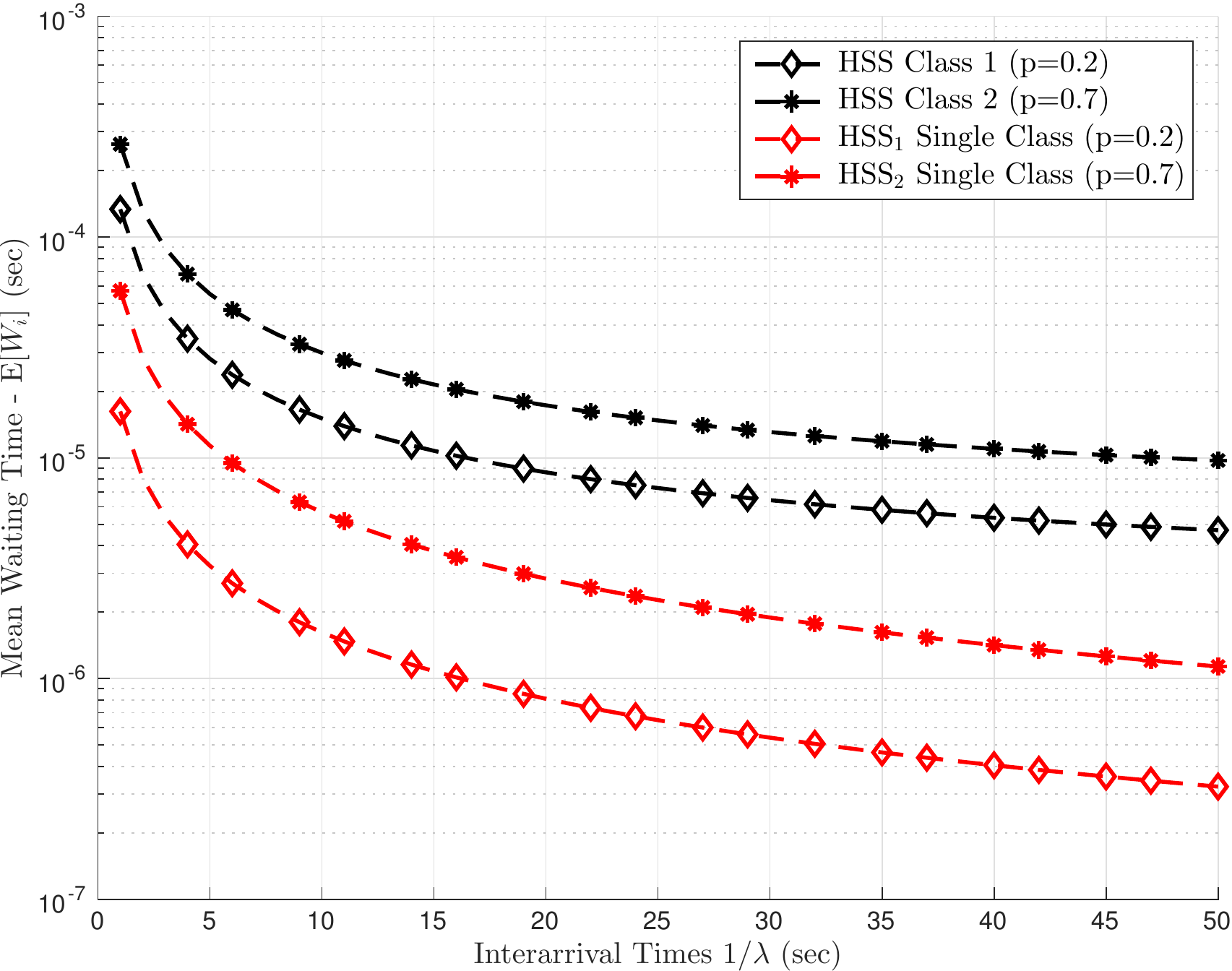}
  		\caption{}
  	\end{subfigure}
  	\vspace{6.5mm}
  	\hspace{15mm}
  	\begin{subfigure}[t]{0.4\textwidth}
  		\centering
  		\includegraphics[width=7.5cm]{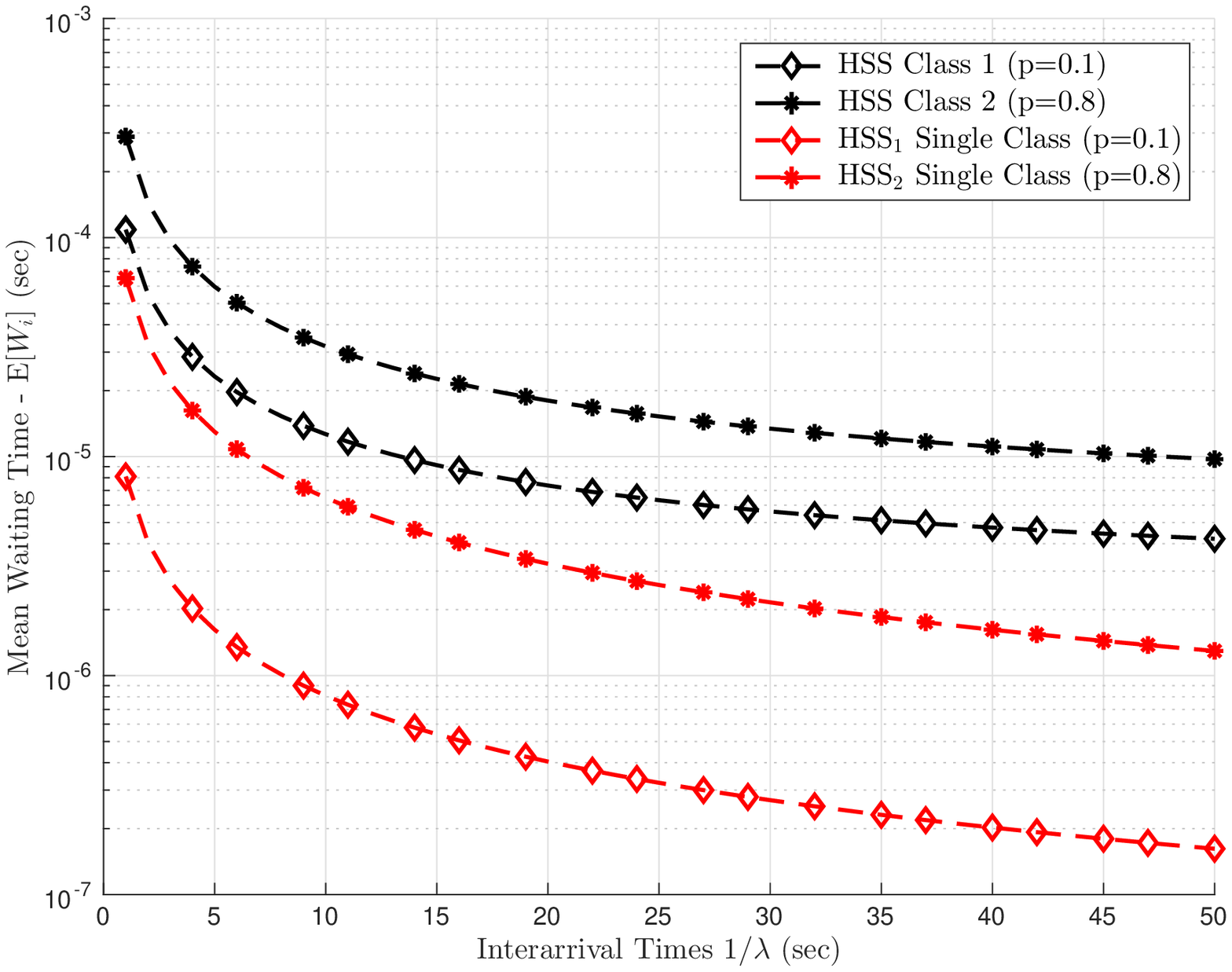}
  		\caption{}
  	\end{subfigure}
  	\caption{Comparison between Single Class and Multi Class schemes in terms of mean waiting time per node for different values of probabilities couples.}
  	\label{fig:singlemulticompare}
  \end{figure*}

\subsection{Multi Class Analysis (BCMP framework)}

In this second part of our performance assessment, we consider the possibility that cIMS requests can be differentiated per class. In fact, many operators often implement their SLAs by separating customers in classes (e.g. gold, silver, bronze) through different HSSs, being these latter designated to manage user profiles. 
Accordingly, it turns to be useful to introduce a variant to the Jackson's framework, known as BCMP networks (the acronym is simply including the initial of the authors). This technique allows taking into account different job classes and different queue disciplines at network nodes. Examples are: FCFS, where the job on top position is served first; and PS where each job in queue receives equal simultaneous service per class.

When there is no need to differentiate classes and to consider queueing policies beyond FCFS, the BCMP reduces to the Jackson framework. The product-form holds again for BCMP networks, and service time distributions (for some queueing policies) must admit a rational Laplace transform \cite{bcmp75}. 
By considering the existence of $l$ classes ($l=1,\dots,L$) of requests, (\ref{eq:balance}) becomes:
\beq
\lambda_{il}=\lambda \cdot p_{0,l} + \sum_{j=1}^{N} \sum_{l=1}^{L} \lambda_{jl} \cdot p_{jl,ir},
\label{eq:multibalance}
\eeq
where: $\lambda_{il}$ is the arrival rate of $l$-th class request to node $i$, $p_{0,l}$ is the probability that arriving requests belong to class $l$, and $p_{jl,ir}$ is the probability that a request belonging to class $l$ and managed by node $j$ acquires the class $r$ and is routed to node $i$. 

Similarly, it is possible to define the mean number of visits $v_{il}$ of a job belonging to the $l-$th class and at node $i$ as:
\beq
v_{il}=p_{0,l} + \sum_{j=1}^{N} \sum_{l=1}^{L} v_{jl} \cdot p_{jl,ir},
\label{eq:visits}
\eeq
with $v_{jl}=\lambda_{jl}/\lambda$. Let us also denote by $k_{il}$ the number of requests belonging to class $l$ at node $i$.
Steady-state probability for BCMP open networks (with load-independent arrival and service rates) admits the same formulation of (\ref{eq:margprobmm1}), but with different $k_i$ values depending on the queueing policy, and amounting to:

 \beqa
 \centering
 \left\{
 \begin{array}{l}
 	{\begin{array}{ll}
 			\hspace{-0.23cm} k_i=\sum_{l=1}^{L}k_{il} , \;\;\;\; \rho_i=\sum_{l=1}^{L} v_{il}\frac{\lambda_l}{\mu_i}   \;\;\;\;\;\;\;\;\; \textnormal{(FCFS nodes)}
 		\end{array}}
 		\\
 		\\
 			\hspace{-0.09cm} k_i=\sum_{l=1}^{L}k_{il} , \;\;\;\; \rho_i=\sum_{l=1}^{L} v_{il}\frac{\lambda_l}{\mu_{il}}   \;\;\;\;\;\;\;\;\; \textnormal{(PS nodes)}.
 	\end{array}
 	\right.
 	\label{eq:bcmpnodes}
 	\eeqa
\newline

\noindent Aimed at evaluating an exemplary multi class scenario, let us consider the case shown in Fig. \ref{fig:multiclass} where two schemes are compared. The uppermost panel shows a scheme implementing the single class scenario with user requests being probabilistically routed towards a specific HSS. On the contrary, the lowermost panel shows a scheme where a single HSS serves two different requests differentiated by means of classes. It is useful to highlight that all HSSs implement a FCFS policy.

\begin{figure*}
	\centering
	\begin{subfigure}[t]{0.4\textwidth}
		\centering
		\includegraphics[width=7.5cm]{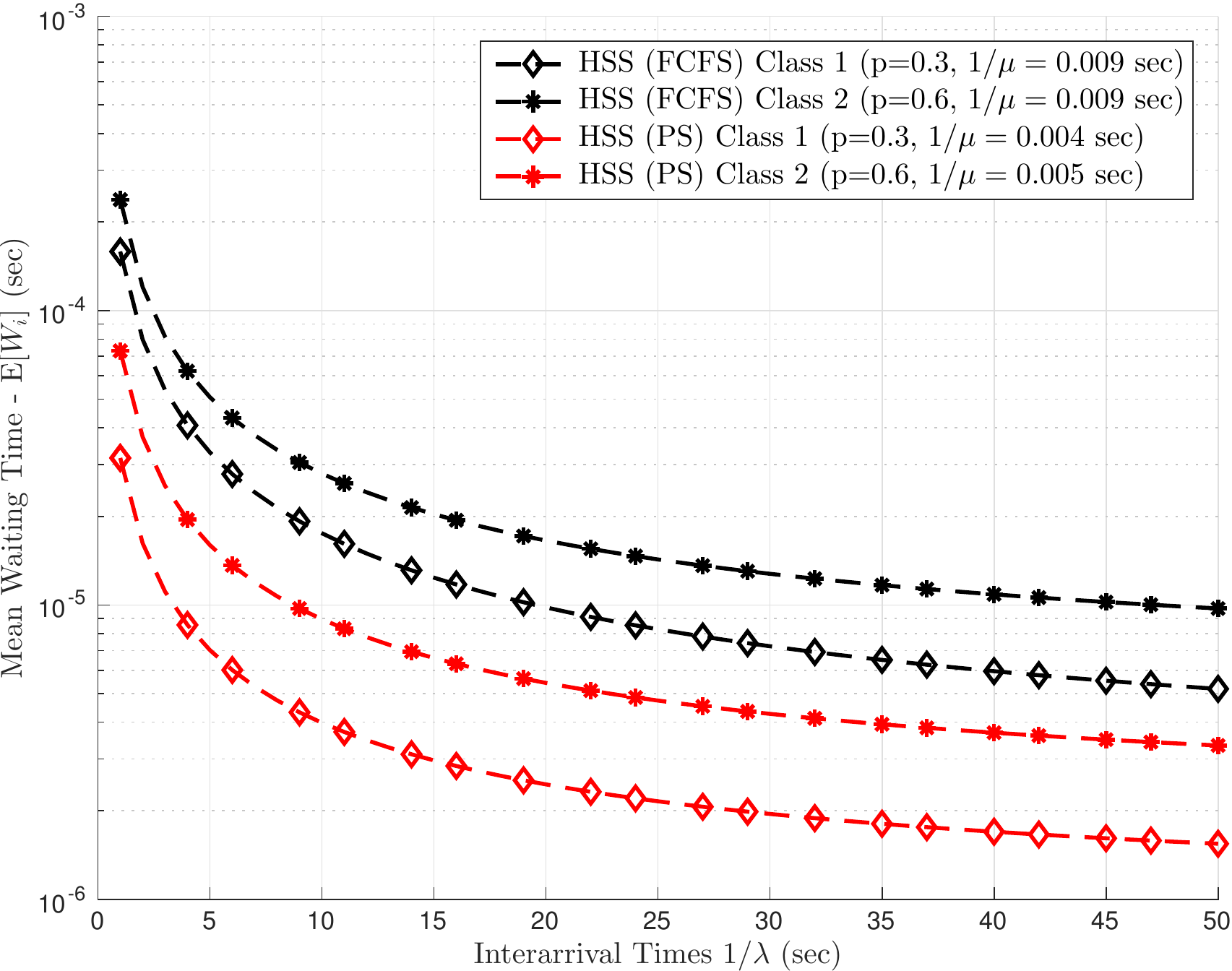}
		\caption{}
	\end{subfigure}
	\vspace{6.5mm}
	\hspace{15mm}
	\begin{subfigure}[t]{0.4\textwidth}
		\centering
		\includegraphics[width=7.5cm]{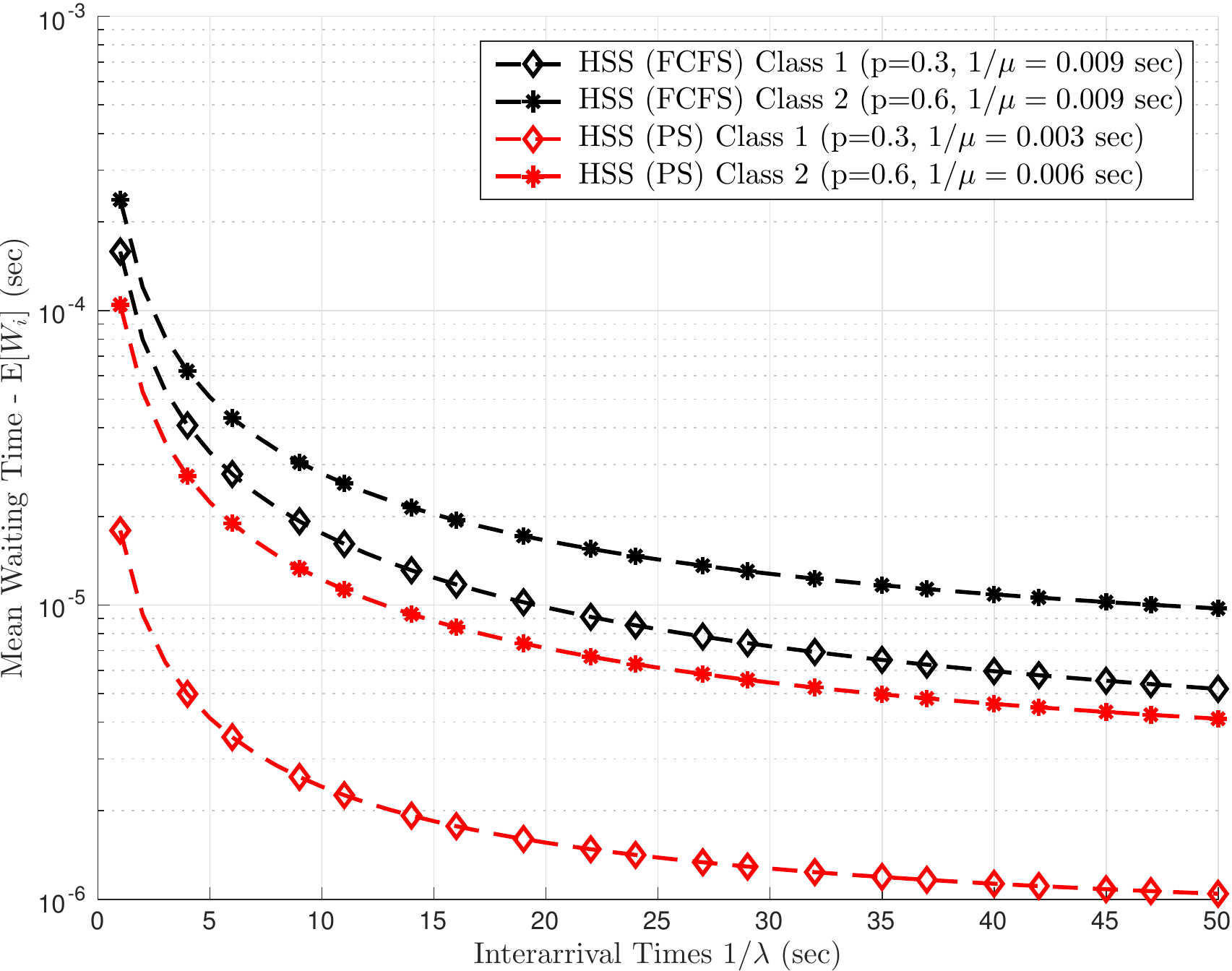}
		\caption{}
	\end{subfigure}
	\begin{subfigure}[t]{0.4\textwidth}
		\centering
		\includegraphics[width=7.5cm]{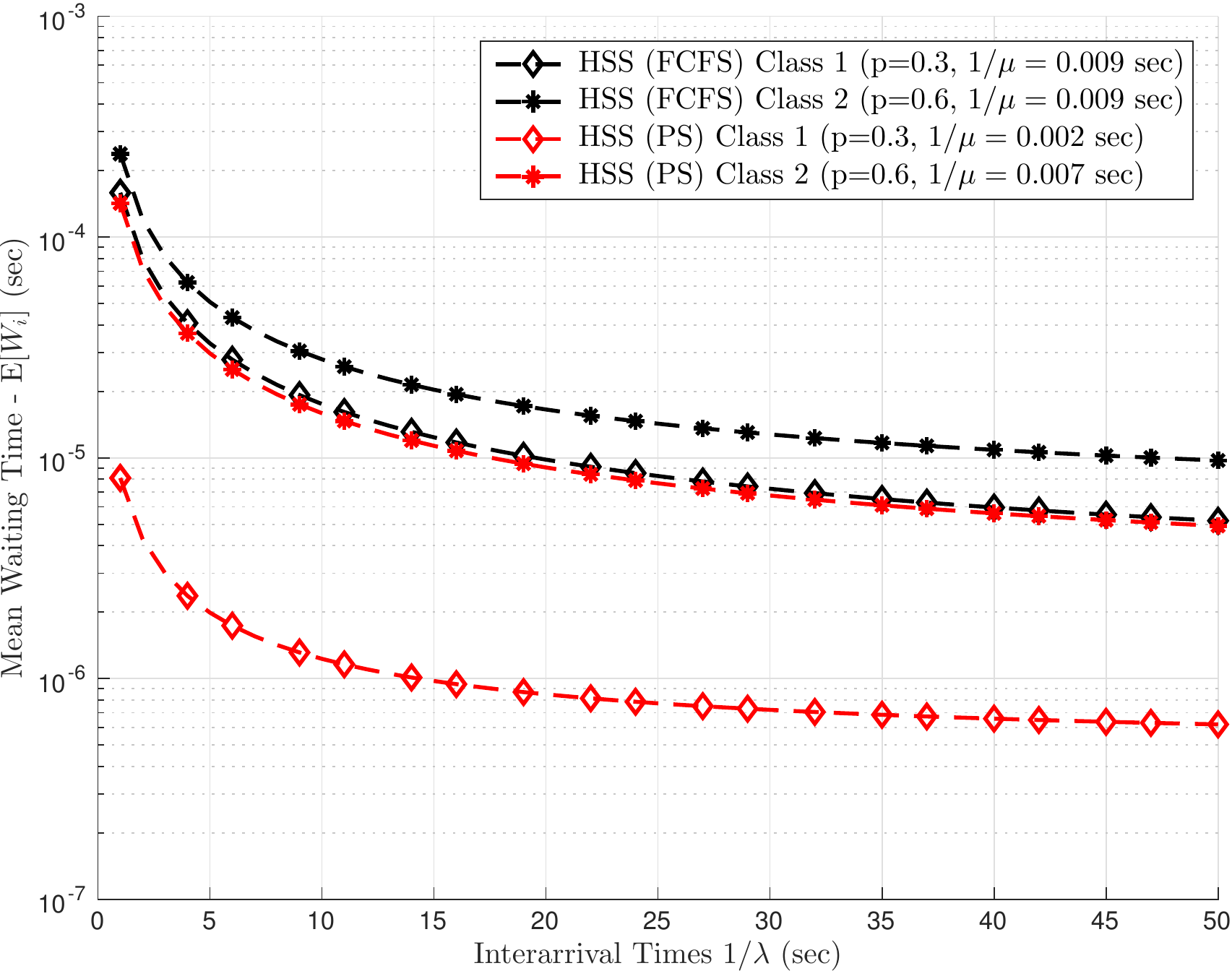}
		\caption{}
	\end{subfigure}
	\vspace{6.5mm}
	\hspace{15mm}
	\begin{subfigure}[t]{0.4\textwidth}
		\centering
		\includegraphics[width=7.5cm]{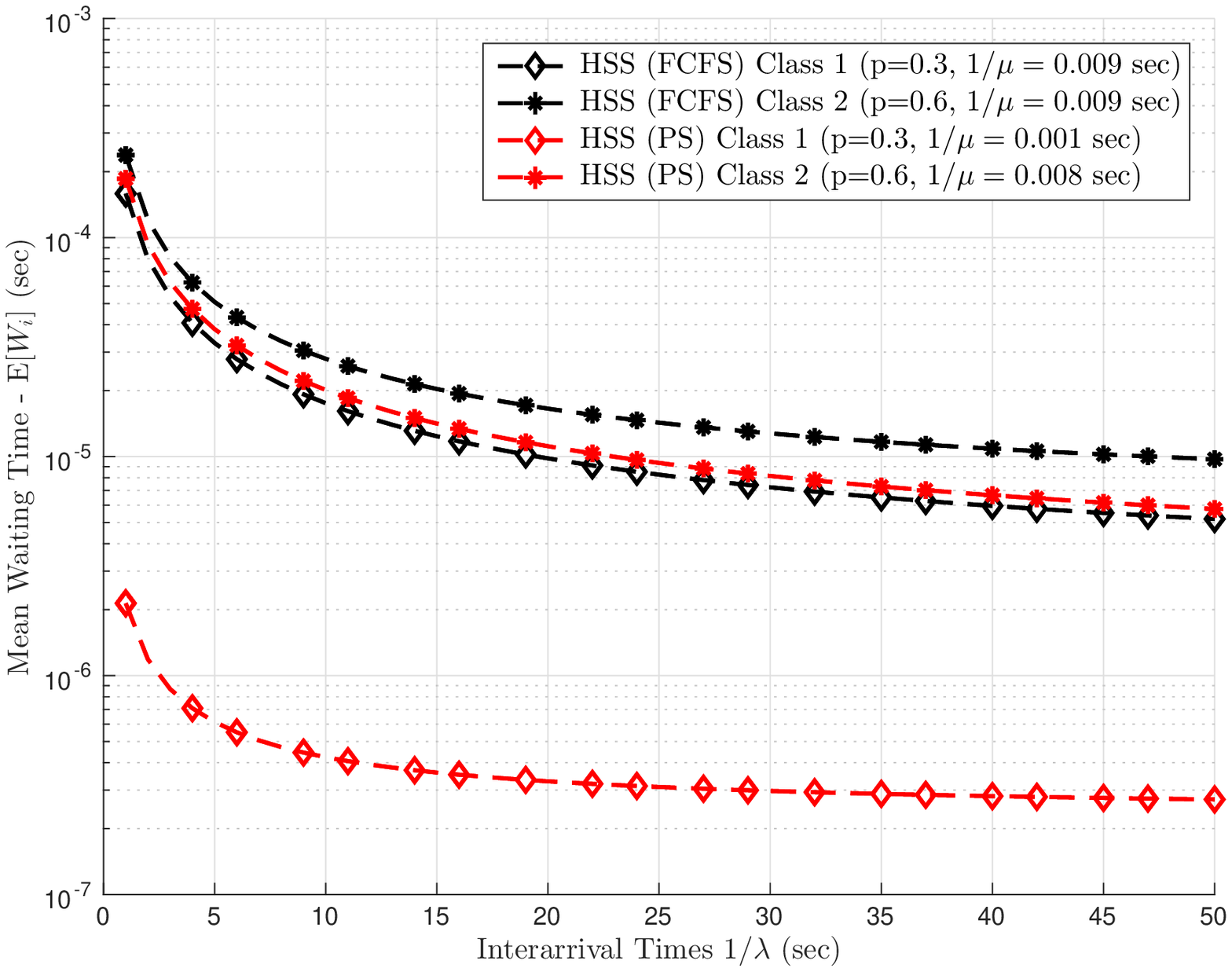}
		\caption{}
	\end{subfigure}
	\caption{Comparison between Multi Class schemes when FCFS and PS queueing policies are considered.}
	\label{fig:psfcfscompare}
\end{figure*}

Let us now compare the two cases when the probability of a request being routed to HSS$_1$ (respectively to HSS$_2$) in the single class scheme equals the probability that the single HSS receives requests belonging to Class $1$ (respectively to Class $2$) in the multi class scheme. The outcomes of this comparison are shown in the panel of Figs. \ref{fig:singlemulticompare}, where the system performance in terms of mean waiting time per HSS node is evaluated, while the service time is kept fixed to $0.009$ for HSS$_1$, HSS$_2$ and HSS. In all figures, red curves refer to the single class scheme (uppermost panel of Fig. (\ref{fig:multiclass})) where HSS$_1$ and HSS$_2$ nodes are queried with two different routing probabilities, whereas the black curves refer to the multi class scheme (lowermost panel of Fig. (\ref{fig:multiclass})). A single HSS node is queried with two probabilities of belonging to class $1$ or class $2$. Performing a pairwise comparison (e.g. HSS Class $1$ vs HSS$_1$), one can notice that $\mathbb{E}[W_i]$ is always lower in the case of single class scheme. Thus, the latter offers more guarantees in terms of latency, and the skew from the multi class scheme becomes more accentuated as the inter-arrival times grow. This is due to the fact that the single class scheme allows exploiting a dedicated HSS resource to manage requests' arrivals. 
On the other hand, a comparison performed between curves belonging to same setting (e.g. HSS Class $1$ vs HSS Class $2$) reveals that, as the probability gap grows (from Fig. \ref{fig:singlemulticompare}(a) to \ref{fig:singlemulticompare}(d)), the mean waiting time gap increases as well. Here, it is interesting to notice that the single class scheme is more adaptive (there is appreciable difference between red curves in the four depicted cases) due to the use of two independent HSS nodes. 
The resulting data could provide useful guidance for a network designer interested at evaluating trade-offs between latency constraints and resource consumption, with the aim to better differentiate SLAs. In practice, the single class setting offers more guarantees than the multi-class scheme in terms of mean waiting time, since it relies on dedicated resources per class. On the contrary, when deploying separate instances becomes costly (consider for instance the license cost per HSS instance), the multi-class solution can be preferable, although at the cost of increased latency.  
As a further analysis, we consider the behavior of a multi-class scheme when the single HSS implements two different queueing policies: FCFS and PS. According to the BCMP framework, the former has to be implemented by considering the same service rate for each class, whereas the latter admits different service rates per class. In line with such indications, we outline some results in the panel of Figs. \ref{fig:psfcfscompare}. Black curves present the multi class case where HSS implements FCFS policy with a fixed service time of $0.009$ and a fixed couple of probabilities per class ($0.3 / 0.6$). On the contrary, red curves refer to the multi-class case where HSS implements PS policy with varying service time per class (the sum amounts to $0.009$) and with the same fixed couple of probabilities per class. 
As a general trend, one can recognize that the PS queueing policy offers better results than FCFS in terms of mean waiting time spent at a node. This is due to a different management of service resources obeying the following behavior: when \textit{r} requests arrive to HSS node, they are simultaneously served with each receiving $1/r$ of the service capacity. Moreover, it is interesting to notice that the PS policy allows a more elastic management than the one offered by FCFS, since it is possible to benefit from the a different allocation of service time per class. In a sense, PS policy exhibits a similar behavior observed in the single class setting with the presence of two separate HSSs. This is due to the possibility of dedicating a ``sliced" service time per class taking into account, at the same time, only one deployed HSS.

\section{Concluding Remarks}
\label{sec:concl}

Today, novel telco architectures (often marketed as $5$G networks) deeply embrace the opportunities offered by virtualized and containerized  environments, since they provide a priceless flexibility in resources managing along with a valuable cost saving. An exemplary case of this marriage is offered by service chains, namely, infrastructures composed of virtualized/containerized nodes traversed in a predetermined fashion to offer a desired service. In line with this nuance, the IP Multimedia Subsystem (IMS) can be interpreted as a particular realization of a service chain. 

In this work we characterize, from a statistical perspective, a service chain represented by a container-based version of the IMS infrastructure, referred to as cIMS. 
We adopt the queueing networks methodology to characterize, as accurately as possible, the mutual interconnections among nodes that, by exhibiting different behaviors, influence the performance metrics of the whole chain (e.g. mean waiting time, mean queue length). 
During this modeling step, we also tackle the case of bulk arrivals at P-CSCF node which leads to a more general version of the
Pollaczek-Khinchin formula.

Then, we adapt and nestle the cIMS model into the so-called open Jackson framework by leveraging the properties of \textit{product-form} networks in order to evaluate the cIMS performance under the hypothesis of single class jobs. Again, we define and solve an optimization problem helpful to highlight the dependencies of cIMS response time from capacity constraints, and to derive the best deployment satisfying a desired cost/resource tradeoff. 

Finally, we introduce the BCMP formalism aimed at extending our assessment to network queues with jobs belonging to different service classes and with nodes implementing different queueing policies. As a result, critical comparisons (based on single/multi class scenarios and on different queueing policies) are proposed, with the aim of pinpointing the optimal cIMS deployments that satisfy the network operators demands.
In this way, the theoretical part is supported by an experimental assessment realized through \textit{Clearwater}, an open source platform that allowed us to deploy a containerized IMS infrastructure, and to derive realistic data useful to strengthen our models.
The obtained results offer useful indications for service providers interested in guaranteeing competitive SLAs across different deployment scenarios, and to limit the resource consumption at the same time. Through the proposed assessment, for instance, a service provider could: $i)$ decide how and where to allocate resources, based on their percentage utilization (e.g. differentiated HSSs); $ii)$ adopt the single class scheme if interested in higher performance in terms of mean waiting time (e.g. for gold class customers); $iii)$ implement a Processor Sharing queueing policy if attracted by a more elastic management (e.g. in case of a multi-tenant architecture).

There are different directions in which the proposed research could be extended in the future. As regards the theoretical part, it will be interesting to analyze the effects of considering redundant instances per cIMS node in order to guarantee the so-called \textit{five nines} or high-availability requirements, which are more than ever required in modern telco deployments. 

From an application level perspective, the proposed characterization may be further tailored across different architectures that exhibit a service chain structure, as often occurs in telco systems. A valuable example is offered by radio access networks, where, traversing a certain number of nodes (e.g. e-node B, Radio Network Controller, etc.) in particular ways could trigger queueing networks issues.

\ifCLASSOPTIONcaptionsoff
  \newpage
\fi

\vspace{200pt}
{\protect\color{black}
\begin{IEEEbiography}[{\includegraphics[width=1in,height=1.15in,clip,keepaspectratio]{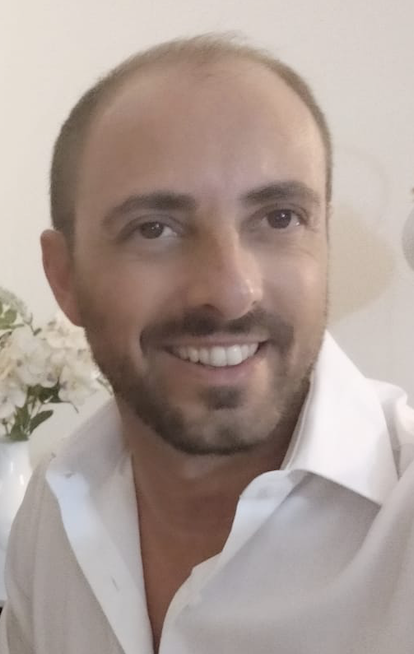}}] 
	{Mario Di Mauro} received the Laurea degree in electronic engineering from the University of Salerno (Italy) in 2005, the M.S. degree in networking from the University of L'Aquila (Italy) jointly with the Telecom Italia Centre in 2006, and the PhD. degree in information engineering in 2018 from University of Salerno.
	He was a Research Engineer with CoRiTel (Research Consortium on Telecommunications, led by Ericsson Laboratory, Italy) and then a Research Fellow with University of Salerno. He has authored several scientific papers, and holds a patent on a telecommunication aid for impaired people. His main fields of interest include: network performance, network security and availability, data analysis for telecommunication infrastructures.
\end{IEEEbiography}
}
\vspace{-200pt}
{\protect\color{black}
\begin{IEEEbiography}[{\includegraphics[width=1in,height=1.15in,clip,keepaspectratio]{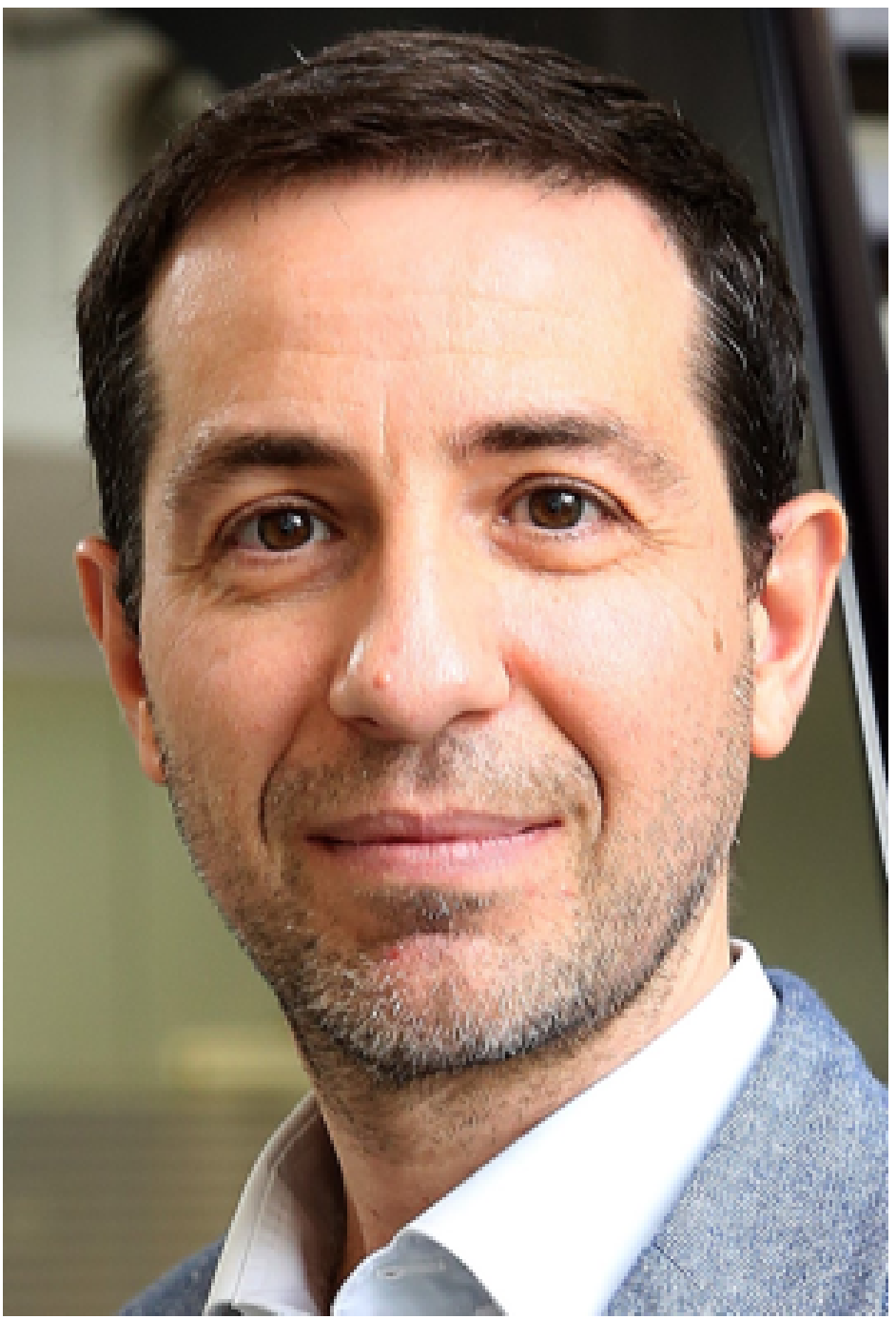}}] 
	{Antonio Liotta} is Professor of Data Science and Intelligent Systems at Edinburgh Napier University, where is coordinating multi-disciplinary programs in Data Science and Artificial Intelligent across the university. He has recently been awarded the prestigious "1000 Talents" fellowship in China, where he is the founding director of the Joint Intellisensing Lab and holds a Visiting Professorship at Shanghai Ocean University. Previously, he was Professor of Data Science and the founding director of the Data Science Research Centre, University of Derby, UK. He was leading all university-wide research, educational, and infrastructure programs in data science and artificial intelligence. His team is at the forefront of influential research in data science and artificial intelligence, specifically in the context of Smart Cities, Internet of Things, and smart sensing. He is renowned for his contributions to miniaturized machine learning, particularly in the context of the Internet of Things. He has led the international team that has recently made a breakthrough in artificial neural networks, using network science to accelerate the training process. Antonio is the Editor-in-Chief of the Springer Internet of Things book series; associate editor of the Journals JNSM, IJNM, JMM, and IF; and editorial board member of 6 more journals.
\end{IEEEbiography}}

\end{document}